\documentclass[reqno,12pt,final]{amsart} 
\usepackage[T1]{fontenc}
\usepackage[utf8]{inputenc}

\usepackage{hyperref}
\usepackage[capitalize]{cleveref}

\usepackage{vmargin}
\setpapersize{A4}

\usepackage{amsmath, mathtools} 

\usepackage{amssymb}      

\usepackage{amscd}      

\usepackage{pgfplots}

\usepackage{amsthm}      

\usepackage{graphicx}

\usepackage[english]{babel}

\usepackage{microtype}

\usepackage{caption}

\usepackage{subcaption}

\pgfplotsset{compat=1.12}

   \theoremstyle{plain}
   \newtheorem{thm}{Theorem}[section]
   \newtheorem{prop}[thm]{Proposition}
   \newtheorem{lem}[thm]{Lemma}  
   \newtheorem{cor}[thm]{Corollary}
   \newtheorem*{thm*}{Theorem}

   \theoremstyle{definition}   
   \newtheorem*{Convention}{Convention}
   
   \newtheorem{defn}[thm]{Definition}
   \newtheorem*{defn*}{Definition}
   \newtheorem{example}[thm]{Example}
   \theoremstyle{remark}
   
   \newtheorem{rem}[thm]{Remark}

\usepackage{mdframed,xcolor}

\definecolor{mybgcolor}{gray}{0.8}
\definecolor{myframecolor}{rgb}{.647,.129,.149}

\mdfdefinestyle{mystyle}{
  usetwoside=false,
  skipabove=0.6em plus 0.8em minus 0.2em,
  skipbelow=0.6em plus 0.8em minus 0.2em,
  innerleftmargin=.25em,
  innerrightmargin=0.25em,
  innertopmargin=0.25em,
  innerbottommargin=0.25em,
  leftmargin=-.75em,
  rightmargin=-0em,
  topline=false,
  rightline=false,
  bottomline=false,
  leftline=false,
  backgroundcolor=mybgcolor,
  splittopskip=0.75em,
  splitbottomskip=0.25em,
  innerleftmargin=0.5em,
  leftline=true,
  linecolor=myframecolor,
  linewidth=0.25em,
}

\newmdenv[style=mystyle]{important}

\usepackage{lipsum}

\usepackage[shortlabels]{enumitem}
  \definecolor{linkcolour}{rgb}{0,0.2,0.6}
  \definecolor{citecolour}{rgb}{0,0.6,0.2}
  \definecolor{urlcolour} {rgb}{0.8,0,0}
\usepackage[draft]{fixme} 

\makeatletter
\g@addto@macro\@floatboxreset\centering
\makeatother

\newcommand\R{\mathbb{R}}

\newcommand\Z{\mathbb{Z}}

\renewcommand{\phi}{\varphi}

\newcommand{\norm}[1]{\Vert #1 \Vert}  
\newcommand{\abs}[1]{\vert #1 \vert}

\newcommand{\del}{\partial}
\newcommand{\sett}[1]{\lbrace #1\rbrace}

 \newcommand{\Int}{\text{Int}}

\bigskip

\makeatletter
\newcommand\InsertTheoremBreak{%
    \@ifstar{\item[\vbox{\null}]}{%
      \begingroup 
      \setlength\itemsep{0pt}%
      \setlength\parsep{0pt}%
       \item[\vbox{\null}]%
      \endgroup%
     }}
\makeatother

\title{Reconstruction of R-regular Objects from Trinary Images}

\author{Helene Matilde Svane}
\author{Andrew du Plessis}

\email{helenesvane@math.au.dk}
\address{Institut for Matematik, Aarhus University, Ny Munkegade, 8000 Aarhus C, Denmark}

\begin{document}

\begin{abstract} We study digital images of $r$-regular objects where a pixel is black if it is completely inside the object, white if it is completely inside the complement of the object, and grey otherwise. We call such images \emph{trinary}. We discuss possible configurations of pixels in trinary images of $r$-regular objects at certain resolutions and propose a method for reconstructing objects from such images. We show that the reconstructed object is close to the original object in Hausdorff norm, {\color{black}and that there is a homeomorphism of $\R^2$ taking the reconstructed set to the original}.
\end{abstract}

\maketitle

\section{Introduction}

The purpose of this paper will be to introduce a way to reconstruct objects from their grey-scale digital images. More specifically, we focus on objects that are small compared to the image resolution and satisfy a certain regularity constraint called \emph{$r$-regularity}. The notion of $r$-regularity was developed independently by Serra \cite{Serra} and Pavlidis \cite{Pavlidis} to describe a class of objects for which reconstruction from digital images preserved certain topological features. They both consider \emph{subset digitisation}, that is, digitisation formed by placing an image grid on top of an object and then colouring an image cell black if its midpoint is on top of the object, and white if the cell midpoint is on top of the complement of the object. This way a binary image is produced, and they consider the set of black cells as the reconstructed set. Serra showed that if the grid is hexagonal and the object satisfies certain constraints, the original and reconstructed sets have the same homotopy, and Pavlidis showed that for a square grid and for certain $r$-regular sets, the set and its reconstruction are homeomorphic. Later on, Stelldinger and Köthe \cite{StelldingerKothe},\cite{SK} argued that the concepts of homotopies or homeomorphisms were not strong enough to fully capture human perception of shape similarity. Instead they proposed two new similarity criterions called \emph{weak} and \emph{strong $r$-similarity}, and showed that under certain conditions, an $r$-regular set and its reconstruction by a square grid are both weakly and strongly $r$-similar. {\color{black}We, too, will consider the notion of weak $r$-similarity in this paper.}

However, Serra, Pavlidis, Stelldinger and Köthe were modelling images using subset digitisation, which outputs a binary image. In contrast to this approach, Latecki et al. \cite{LateckiConradGross} modelled an image by requiring that the intensity in each pixel be a monotonic function of the fraction of that object covered by that pixel. This way they seek to model a pixel intensity as the light intensity measured by a sensor in the middle of the pixel, and the result is a grey-level image much like the ones obtained in real situations. They show that after applying any threshold to such an image of an $r$-regular object with certain constraints, the set of black pixels has the same homotopy type as the original object and, in the case where the original object is a manifold with boundary, the two are even homeomorphic. They also conjecture that all $r$-regular objects are manifolds with boundary. This was later proven by Duarte and Torres in \cite{DTs}.

We will model our images in the same way as Latecki et al. did, namely by requiring each pixel intensity to be a monotonic function of the fraction of the pixel covered by the object. In contrast to the above reconstruction approaches, we do not wish to use a set of pixels as our reconstructed set, but rather to construct a new set with smooth boundary that we may then use as the reconstruction. Also in contrast to the above, we will not consider binary images, but keep the information stored in the grey values in our endeavour to make a more precise reconstruction.

When reconstruction, one should decide which properties one wishes the reconstructed object to share with the original one. Should the reconstructed set have the same topological features as the original one? Should the reconstructed set be close to the original one? Should a digitisation of the reconstructed set yield the same image as the original set? Should the reconstructed set be $r'$-regular for some $r'$ close to $r$? Though all of these comparison criteria are interesting to work with and  an ideal reconstruction should satisfy them all, it is hard to construct such a set. In this paper, we will therefore focus on constructing a set that is close to the original one in Hausdorff distance (which will be introduced in the following), has a smooth boundary{\color{black}, and is homeomorphic to the original set. This means that we show that our reconstructed set and the original are weakly $r$-similar in the sense of \cite{SK}}.

\section{Basic definitions and theorems about \texorpdfstring{$r$}{r}-regular sets}\label{SecRBasics}
Let us start by establishing some terminology.
Let $X\subset \R^2$ be a set. We will denote the closure of $X$ by $\overline{X}$, the interior of $X$ by $\text{Int}(X)$ and the boundary of $X$ by $\del X$. The complement $\R^2\backslash X$ will be denoted by $X^C$. The set $X$ is compact if and only if $X$ is closed and bounded.

The Euclidean distance between two points $x$ and $y$ in $\R^2$ will be denoted by $d(x,y)$ or, occasionally, by $\norm{x-y}$.

For an $s>0$, we let $B_s(x)=\sett{y\in \R^2\mid d(x,y)<s}$ be the open ball with centre $x$ and radius $s$. For a line segment $L$ we will denote the length of $L$ by $\abs{L}$.

{\color{black}A part of the }
 goal will be to construct a set from a digital image whose boundary is close to the boundary of the original set. The intuitive concept of closeness between two sets is captured by the Hausdorff distance: For $X,Y\subseteq \R^2$, the Hausdorff distance $d_H$ between $X$ and $Y$ is given by
\begin{equation*}
d_H(X,Y)=\max\sett{\sup_{x\in X} \inf_{y\in Y}d(x,y), \sup_{y\in Y} \inf_{x\in X}d(x,y)}.
\end{equation*}
The set of compact sets of $\R^2$ equipped with the Hausdorff metric is a complete metric space.

The digital images that we will be working with in this paper are formed in the following way:

\begin{defn}
Let $X\subseteq \R^2$ be a set and $d\Z^2\subseteq\R^2$ a grid
with side length $d$. To each grid square $C$, we assign an \emph{intensity} $\lambda$ given by
\begin{equation*}
  \lambda=\phi\left(\frac{\text{area}(X\cap C)}{d^2}\right)\in[0,1],
\end{equation*}
where $\phi:[0,1]\to [0,1]$ is a monotonic function with $\phi(0)=0$, $\phi(1)=1$ and $\phi((0,1))\subseteq (0,1)$.

The \emph{digitisation} of $X$ is the matrix of
intensities. We will visualise it as the collection of pixels of side length $d$, each coloured a shade of grey corresponding to the value of $\lambda$.

Let $V(X)$ denote \emph{the black pixels} of this
digitisation of $X$. We will sometimes refer to $V(X)$ as the black
digitisation pixels of $X$.
\end{defn}

To make sure that the objects in the images we are considering are not arbitrarily strange, we will follow in the footsteps of previous approaches and only consider $r$-regular sets:

\begin{defn}
Let $r>0$. A closed set $X\subseteq \R^n$ is said to be \emph{$r$-regular} if for each $x\in\del X$ there exists two $r$-balls $B_r(x_b)\subseteq X$ and $B_r(x_w)\subseteq X^C$ such that $\overline{B_r(x_b)}\cap \overline{B_r(x_w)}=\sett{x}$, see \cref{FigR-reg}.
\end{defn}

\begin{figure}
\includegraphics[scale=1]{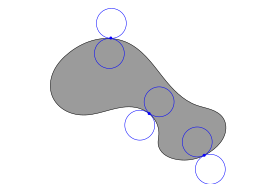}
\caption{An $r$-regular set $X$ is a set where each boundary point belongs to both the boundary of an $r$-ball contained in $X$ and the boundary of an $r$-ball contained in $X^C$}
\label{FigR-reg}
\end{figure}

In general, we believe that a reconstruction can be made more accurately by taking the intensities of the grey pixels into account, and we are currently working on this idea. However, in this paper we restrict ourselves to looking at images where each pixel is considered to be either black, grey or white, without taking the exact intensities of the grey pixels into account:

\begin{defn}
A \emph{trinary} digital image is a digital image where the intensities of all grey pixels are set to 0.5.
\end{defn}

These trinary images will be our main interest in this paper. Note that the colour of a pixel (black, grey or white) does not depend on the monotonic function $\phi$ used for calculating the pixel intensities - in fact, a pixel in a trinary image of an object $X$ is black if it is contained in $X$, white if it is contained in $X^C$ and grey if $\del X$ passes through it.

When we make the digital image of an $r$-regular object by a lattice $d\Z^2$, we can in general not be certain that there are any black or white pixels in the image - for instance, if $d$ is large compared to $r$, all pixels could contain an $r$-ball, which would mean that the image would be all grey. Since we cannot hope to make a very good reconstruction in this case, we will put a restriction on the relationship between the $r$ and $d$:

\begin{Convention}
Throughout the following, we assume that $X$ is a bounded $r$-regular set and that $d\sqrt{2}<r$. We also assume that $\del X$ does not pass through  a pixel corner.
\end{Convention}

Note that the boundedness condition on $X$ implies that $X$ is compact.

Pavlidis \cite{Pavlidis} defines a grid $d\Z^2$ and a set $X$ to be \emph{compatible} if $X$ is $r$-regular with $d\sqrt{2}<r$. With this restriction, since $d\sqrt{2}$ is the diameter of a pixel, each black $r$-ball contains the pixel that its centre belongs to, meaning that each black $r$-ball is centered in a black pixel. Similarly the centre of each white $r$-ball is contained in a white pixel. This means that for each component of $X$ yields at least one black pixel, and each component of $X^C$ yields at least one white pixel. Latecki et al. showed that for a compatible grid $d\Z^2$ and set $X$, the set $V(X)$ of black pixels is homeomorphic to $X$. Hence $X$ and $V(X)$ have the same topological features. Furthermore, the above conditions ensure that we do not get too large grey areas, as will be clear in the following section. We will only concern ourselves with images that capture all of the objects photographed, and not just a part of them.

{\color{black}Let us introduce the notion of \emph{weak $r$-similarity}, as introduced in \cite{SK},\cite{SLS}.

\begin{defn}
Let $A, B\subseteq \R^2$ be bounded sets and $r>0$. We call $A$ and $B$ weakly $r$-similar if there exists a homeomorphism $f:\R^2\to\R^2$ such that $x\in A\iff f(x)\in B$ and the Hausdorff distance between the set boundaries satisfies $d_H(\del A, \del B)<r$
\end{defn}

The overall purpose of this paper will be to show the following:
\begin{thm}\label{ThmMainResult}
Let $I$ be a digital image of an $r$-regular set $X$ by a lattice $d\Z^2$ with $d\sqrt{2}<r$. We may construct an object $\Gamma$ from $I$ such that $\Gamma$ and $X$ are weakly $d$-similar, where $d$ is the pixel side length.
\end{thm}

We believe that the above result may be strengthened to prove strong $d+\varepsilon$-similarity between the two for a suitable $\varepsilon$, but such a result is beyond the scope of this paper. 

A large part of the proof of Theorem \ref{ThmMainResult} will be to prove the following:
}

\begin{thm}
Let $I$ be a digital image of an $r$-regular set by a lattice $d\Z^2$ with $d\sqrt{2}<r$. We may construct an object $\Gamma$ from $I$ such that $d_H(\del \Gamma,\del X)<d$, where $d_H$ is the Hausdorff distance.
\end{thm}

To start working with $r$-regular sets, we first sum up some basic statements about them:

\begin{prop}[Tang Christensen and du Plessis, \cite{TC}, Proposition A.1]\label{PropEquivalentRregDefinitions}
Let $A\subseteq \R^n$ be a closed set and $r>0$. Then the following are equivalent:
\begin{enumerate}
\item At any point $x\in\del X$ there exist two closed $r$-balls $B_r\subseteq A$ and $B_r'\subseteq\overline{A^C}$ such that $B_r\cap B_r'=\sett{x}$.
\item The sets $A$ and $\overline{A^C}$ are equal to unions of closed $r$-balls.
\end{enumerate}
\end{prop}

\begin{defn}
  For $\delta>0$, we denote the $\delta$-tubular neighbourhood of
  $\del X$ in $\R^2$ by
  $N_\delta=\sett{x\in\R^2\mid d(x,\del X)<\delta}$.
\end{defn}

\begin{lem}[Duarte \& Torres, \cite{DTs}, Lemma 5]\label{LemProjection}
  Let $X$ be an $r$-regular set. For each $x\in N_{r}$ there is a
  unique point $\pi(x)\in\del X$ such that $d(x,\del
  X)=d(\pi(x),x)$. Hence there is a well-defined projection
  $\pi:N_{r}\to\del X$.
\end{lem}

\begin{thm}[Duarte and Torres, \cite{DTr}]
The projection map $\pi: N_r\to\del X$ is continuous.
\end{thm}

%
%

Another important fact that we will be using heavily is the following:

There is a retraction $\rho_{X^C}:N_r\to X^C\cup\del X$ (that we will
sometimes just denote by $\rho$) defined by
\begin{equation*}
  \rho_{X^C}(x)=\begin{cases}
    x & \text{if $x\in X^C\cup\del X$},
    \\
    \pi(x) & \text{otherwise},
  \end{cases}
\end{equation*}
and likewise a retraction $\rho_{X}:N_r\to X$ defined by
\begin{equation*}
  \rho_{X}(x)=\begin{cases}
    x & \text{if $x\in X$},
    \\
    \pi(x) & \text{otherwise}.
  \end{cases}
\end{equation*}
These retractions will prove to be crucial in later arguments, since
they have some nice properties.

We now state some results about $\rho=\rho_{X^C}$. However, the
similar results for $\rho_X$ also hold.

\begin{prop}[Stelldinger et al., \cite{SLS}] Let $x, y\in X^C$ with
  $d(x,y)<2r$ and let $L\subseteq\R^n$ be the line segment between
  them. Then
  \begin{enumerate}[(i)]
  \item The line segment $L$ is a subset of $ X^C\cup N_r$, and
    $\rho\vert_L$ is injective,
  \item For $s<r$ and $B_s$ any $s$-ball containing $x$ and $y$,
    $\rho(L)$ is a subset of $B_s$.
  \end{enumerate}
\end{prop}

\begin{defn}
  Let $L\subseteq\R^n$ be a closed line segment of length
  $\vert L\vert<2r$. Then the $r$-spindle $S(L,r)$ around $L$ is the intersection
  of all closed balls of radius $r$ whose boundaries contain both endpoints of $L$. If $x$ and $y$ are the endpoints of $L$, we will sometimes write $S(x,y,r)$ in stead of $S(L,r)$.
\end{defn}

\begin{lem}[du Plessis, A.20 \cite{TC}]\label{LemSpindleWidth}
Let $L$ be a closed line segment in $\R^2$ of length $\abs{L}<2r$. Then the maximal distance from a point in the $r$-spindle $S(L,r)$ to $L$ is $r-\sqrt{r^2-\frac{L^2}{4}}$.
\end{lem}

\begin{lem}[du Plessis, A.13 \cite{TC}]\label{LemSpindleIsIntersection}
  Let $L\subseteq\R^n$ be a closed line segment of length
  $\vert L\vert<2r$. Then the $r$-spindle $S(L,r)$ is the intersection
  of all balls of radius at most $r$ that contain $L$.
\end{lem}

\begin{cor}[du Plessis, A.16 \cite{TC}]\label{CorPathInSpindle}
  Let $x, y\in X^C$ with $d(x,y)<2r$ and let $L\subseteq\R^n$ be the
  line segment between them. Then $\rho(L)$ is a subset of the
  $r$-spindle $S(L,r)$.
\end{cor}

\begin{rem}
Since $\pi=\rho_X\circ\rho_{X^C}=\rho_{X^C}\circ\rho_X$, the above
corollary is also true for $\pi$.
\end{rem}

\section{Impossible configurations at a resolution satisfying \texorpdfstring{$d\sqrt{2}<r$}{dsqrt(2)<r}}\label{SecConfigurations}
Before we start reconstructing the original $r$-regular object, we need to discuss which configurations of $3\times 3$ pixels of grey, black and white pixels can occur in the digital image of an $r$-regular object by a lattice $d\Z^2$ where $d\sqrt{2}<r$. We can make a computer put together all possible configurations of $3\times 3$ pixels by telling it that the only possible configurations of $2\times 2$ pixels are the ones in Figure \ref{Fig2x2Configurations}, up to rotation and interchanging of black and white. We can then make a MatLab programme that combines these configurations in all possible ways.
If we do this, we get (up to rotation, mirroring and switching of black and white pixels) the configurations in Figure \ref{FigMany3x3Configurations}.

\begin{figure}
\includegraphics[scale=0.3]{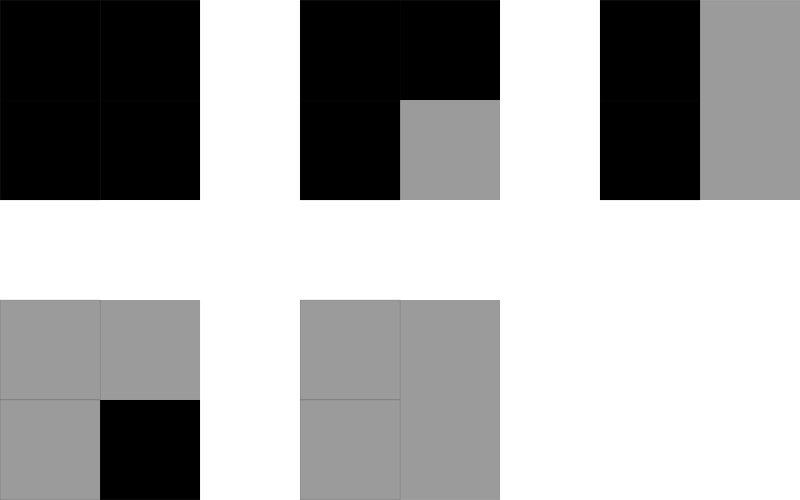}
\caption{The only possible configurations of $2\times2$ pixels, up to rotation and switching of black and white. Note that we have used Lemma \ref{LemSkalfarvessorte1}, part ii), which is stated below.}
\label{Fig2x2Configurations}
\end{figure}

\begin{figure}
\includegraphics[scale=0.7]{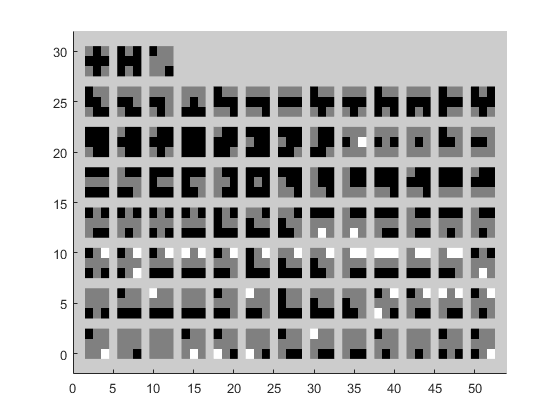}
\caption{All possible combinations of the allowed $2\times2$ pixel configurations, up to rotation, mirroring and interchanging of black and white pixels.}
\label{FigMany3x3Configurations}
\end{figure}

Note that not all these configurations can occur in the image of some $r$-regular object by a lattice $(d\Z)^2$ with $d\sqrt{2}<r$. We would like to remove configurations that do not occur from the list in Figure \ref{FigMany3x3Configurations}. To do so, we need to prove a series of lemmas. Their proofs are mainly geometric and rather technical, so we will put them in the appendix instead of presenting them here.

First of all, let us start with a definition, borrowed from Pavlidis' book \cite{Pavlidis}.

\begin{defn}
Two pixels are \emph{direct neighbours} (abbreviated \emph{d-neighbours}) if the respective cells share a side. Two pixels are \emph{indirect neighbours} (abbreviated \emph{i-neighbours}) if those cells touch only at a corner. The term \emph{neighbour} denotes either type.
\end{defn}

In the following lemmas, we will only be considering pixel configurations in images of $r$-regular objects by lattices $d\Z^2$ with $d\sqrt{2}<r$ according to our convention, but for brevity we will omit this requirement from the lemma statements. 

\begin{lem}\label{LemTwoIntersections}
  Consider four pixels as in Figure \ref{FigToSkaeringer}. Suppose
  $\del X$ intersects the edge between the two pixels $B$ and $C$ more
  than once. Then one of the pixels $A$ and $D$ is black, and the
  other one is white.
  The same result is true if $L$ is tangent to $\del X$ in a point.
\end{lem}

\begin{figure}
  \includegraphics[height=5cm]{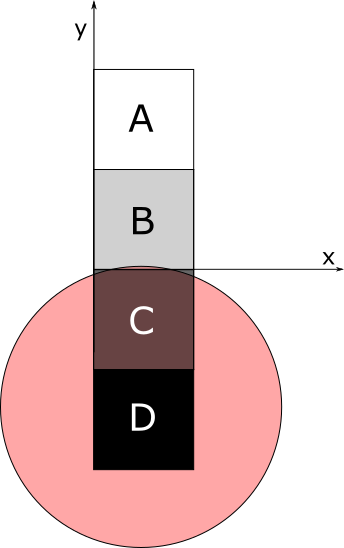}
  \caption{Consider four pixels as in the figure, where the boundary
    $\del X$ intersects the edge between the pixels $B$ and $C$
    twice. The proof consists of showing that there must be two $\sqrt{2}d$-balls with
    centres in $A$ and $D$, respectively, and that one of the balls is
    black and the other one white.}
  \label{FigToSkaeringer}
\end{figure}

\begin{figure}
\centering
\includegraphics[scale=0.45]{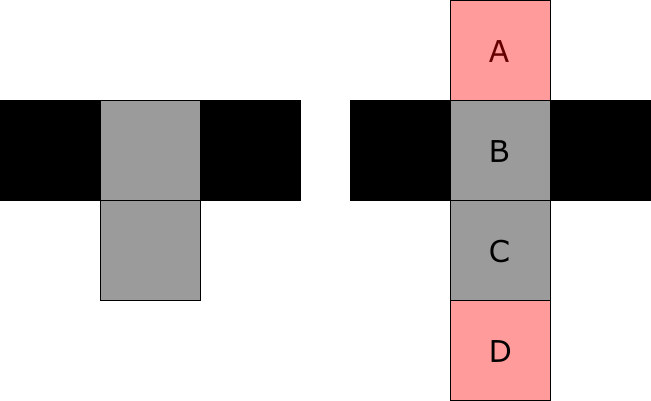}
\caption{If a configuration as the left one occurs in a digital image of an $r$-regular object with $d\sqrt{2}<r$, then pixel $A$ must be black, and pixel $D$ must be white.}
\label{FigYieldsTwoIntersections}
\end{figure}

\begin{lem}\label{LemYieldsTwoIntersections}
In a configuration as the one in Figure \ref{FigYieldsTwoIntersections} left, the pixel named $A$ in Figure \ref{FigYieldsTwoIntersections} right must be black, and the the pixel named $D$ must be white.
\end{lem}

\begin{lem}\label{LemNo9greys}
  Consider a configuration of $3\times3$ pixels with the middle one grey. Then one of its 8 neighbour pixels is not grey.
\end{lem}

\begin{figure}
  \includegraphics[scale=0.5]{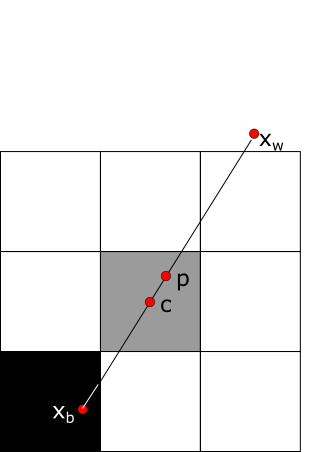}
  \caption{In Theorem \ref{Thm5greys}, we consider 9 pixels, of which
    the middle one is grey, as in the figure.}
  \label{Fig5greysAdd}
\end{figure}

\begin{lem}\label{Thm5greys}
  Consider a configuration of $3\times3$ pixels as in Figure
  \ref{Fig5greysAdd}, where the middle one is grey and has centre
  $c$. Let $p=\pi(c)$ be the point of $\del X$ that is nearest to $c$, and
  suppose the centre of the black ball $B_r(x_b)$ tangent to $\del X$
  at $p$ is closer than the centre of the white ball $B_r(x_w)$
  tangent to $\del X$ at $c$, and that $x_b$ belongs to the lower left
  pixel (which is hence black).

  Then the upper right pixel is white.
\end{lem}

\begin{figure}
\includegraphics[scale=0.4]{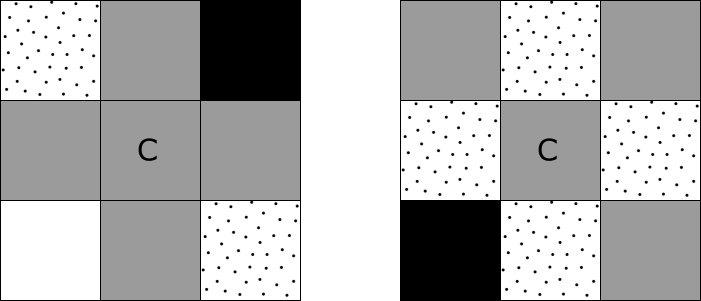}
\caption{If a grey pixel has four grey d-neighbours as in the left figure, it must have a black and a white i-neighbour sitting diagonally across from each other. Equivalently, if a grey pixel does not have a black and a white i-neighbour sitting across from each other as in the right figure, it cannot have four grey d-neighbours.}
\label{FigRemark}
\end{figure}

\begin{rem}\label{Rem3x3Pix}
  Theorem \ref{Thm5greys} and Lemma \ref{LemNo9greys} combined tell us
  that a grey pixel $C$ with four grey d-neighbours must always have a black and a white i-neighbour whose common vertices with $C$ sit diagonally across from each other, see Figure \ref{FigRemark}, left.
  Equivalently, if a grey pixel $C$ does not have a black and a white neighbour sitting opposite of each other, then at least one of its d-neighbours must not be grey.
\end{rem}

\begin{figure}
  \centering
  \begin{minipage}[b]{0.3\textwidth}
    \centering
    \includegraphics[scale=0.4]{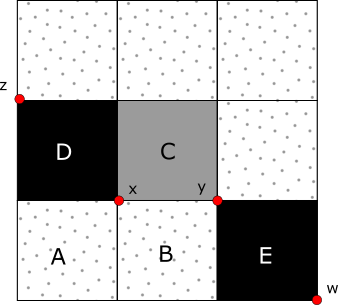}
    \subcaption{}
    \label{FigSkalfarvessorte1}
  \end{minipage}
  \begin{minipage}[b]{0.3\textwidth}
    \centering
    \includegraphics[scale=0.4]{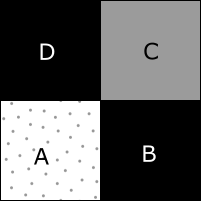}
    \subcaption{}
    \label{FigSkalfarvessorte3a}
  \end{minipage}
  \begin{minipage}[b]{0.3\textwidth}
    \centering
    \includegraphics[scale=0.4]{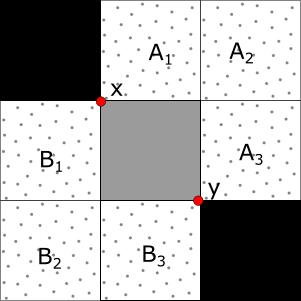}
    \subcaption{}
    \label{FigSkalfarvessorte2}
  \end{minipage}
  \caption{Consider configurations of two black and a grey pixel (we
    do not assume anything about the colour of the dotted pixels). If
    a grey and two black pixels sit in a configuration
    as in Figure \ref{FigSkalfarvessorte1} or
    \ref{FigSkalfarvessorte3a}, then the pixels $A$ and $B$ must also
    be black. If a grey and two black (or two
    white) pixels sit in a configuration as in Figure
    \ref{FigSkalfarvessorte2}, then either the pixels $A_1$, $A_2$,
    $A_3$ are all black, or the pixels $B_1$,
    $B_2$, $B_3$ are all black.}
\end{figure}

\begin{lem}\label{LemSkalfarvessorte1}\label{LemSkalfarvessorte2}
The following holds:
  \begin{enumerate}[(i)]
  \item Consider $2\times3$ pixels as in the lower part of Figure
    \ref{FigSkalfarvessorte1} with the grey and black pixels placed
    relatively to each other as in the figure. Then pixels $A$ and $B$
    must necessarily be black.

  \item Consider $2\times 2$ pixels as in Figure
    \ref{FigSkalfarvessorte3a}, with the grey and black pixels placed
    relatively to each other as in the figure. Then $A$ must necessarily
    be black.

  \item Consider $3\times 3$ pixels as in Figure
    \ref{FigSkalfarvessorte2}, with the grey and black pixels placed
    relatively to each other as in the figure. Then either the pixels
    $A_1$, $A_2$, $A_3$ are all black, or the pixels $B_1$, $B_2$,
    $B_3$ are all black.
  \end{enumerate}
  \medskip
  \noindent
  The similar result is also true if we replace the black pixels with
  white ones in the figures.
\end{lem}

\begin{figure}
 \includegraphics[height=5cm]{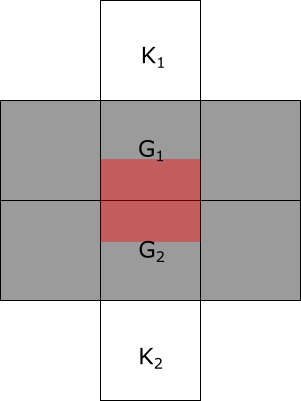}
 \caption{We are considering 6 grey pixels in a $2\times3$
    combination, and we have shown that $\del X\cap(G_1\cup G_2)$
    belongs to the red set in the figure, and that one of the pixels
    $K_1$, $K_2$ must be black, and the other one white.}
 \label{Fig6greys}
\end{figure}

\begin{lem}\label{Lem2x3grey}
  Suppose we have a configuration of 6 grey pixels as in Figure
  \ref{Fig6greys}, with pixels $G_1$, $G_2$, $K_1$ and $K_2$ as in the
  figure. Then the following holds:
  \begin{enumerate}[(i)]
  \item One of the pixels $K_1$, $K_2$ must be black, and the other
    one white,
  \item The set $\del X\cap(G_1\cup G_2)$ belongs to the set of points
    in $G_1\cup G_2$ that are no further than $(\sqrt{2}-1)d$ from the
    common edge of $G_1\cup G_2$ (i.e. the red set in the figure).
  \end{enumerate}
\end{lem}

\begin{lem}\label{LemImpossible2x3}
A configuration as the one in Figure \ref{Fig2x3Impossible} left cannot occur.
\end{lem}

\begin{figure}
\includegraphics[scale=0.45]{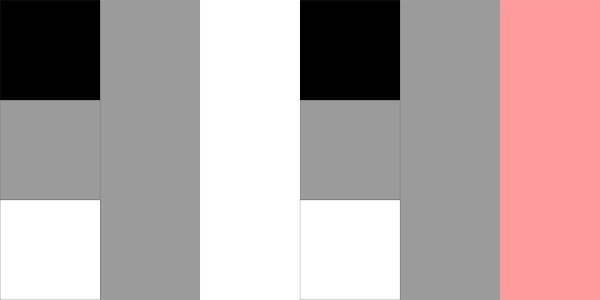}
\caption{The configuration to the left cannot occur in the digital image of an $r$-regular object with $d\sqrt{2}<r$, for if it did, one of the pixels that are coloured red to the right would have a colour that would not be compatible with any legal configuration.}
\label{Fig2x3Impossible}
\end{figure}

\begin{figure}
\includegraphics[scale=0.4]{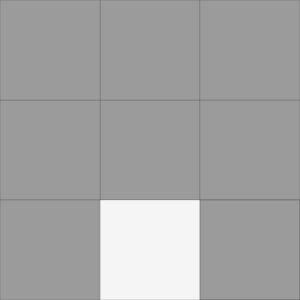}
\caption{The configuration in this figure cannot occur in the image of an $r$-regular object with $d\sqrt{2}<r$.}
\label{FigImpossible3x3A}
\end{figure}

\begin{lem}\label{LemTrickyConfiguration}
A configuration as the one in Figure \ref{FigImpossible3x3A}, left cannot occur.
\end{lem}

\begin{thm}\label{ThmConfigurationsUsualResolution}
Up to rotation, mirroring and interchanging of black and white, any $3\times 3$ configuraton of pixels is one of those shown in Figure \ref{FigConfigurationsUsualResolution}.
\end{thm}

\begin{figure}
\includegraphics[scale=0.4]{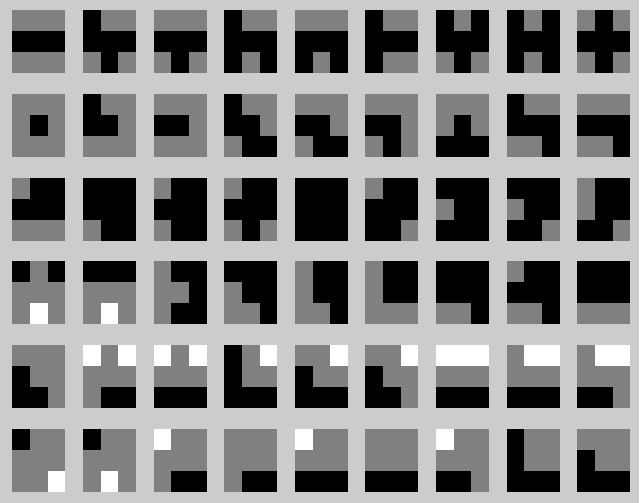}
\caption{Up to rotation, mirroring and switching of black and white colours, these are the only $3\times 3$-configurations that can occur in the digital image of an $r$-regular object by a lattice $(d\Z)^2$ with $d\sqrt{2}<r$.}
\label{FigConfigurationsUsualResolution}
\end{figure}

In the following, it will also be useful to know which $4\times4$-configurations with a centre of $2\times 2$ grey pixels that may occur in a digital image of an $r$-regular object by a lattice $d\Z^2$. We may have a computer find these by combining all possible $3\times 3$-configurations from Figure \ref{FigConfigurationsUsualResolution}, together with the rotations, mirror images and inverses of these configurations.  After removing configurations that violate Lemma \ref{Lem2x3grey}, this yields the configurations in Figure \ref{FigKonfigurationer} (up to rotations, mirror images and interchanging of black and white). We aim to remove configurations from this list if they cannot occur in a digital image like the ones we are considering.

\begin{figure}
\includegraphics[scale=0.5]{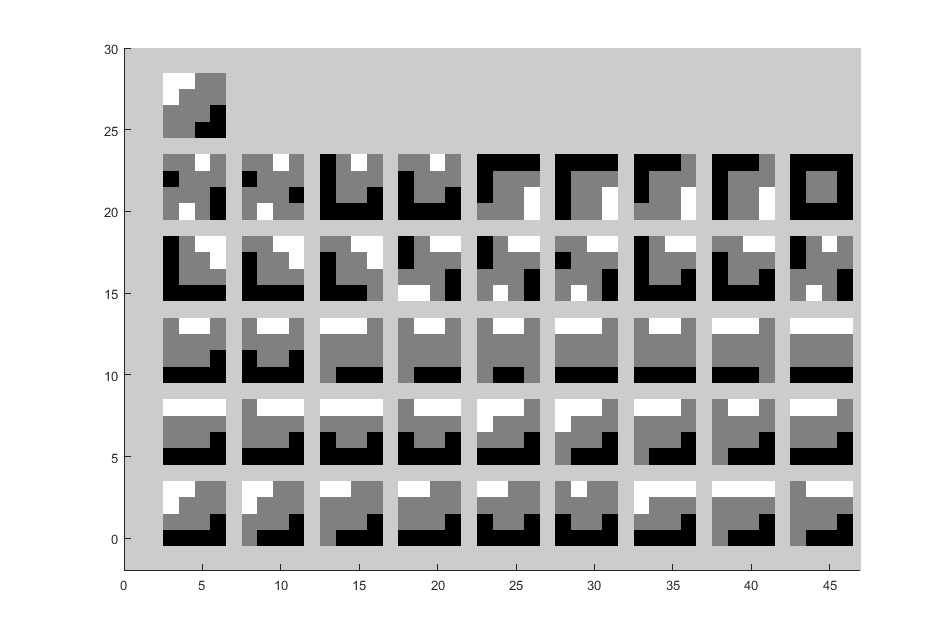}
\caption{These are all possible combinations of the configurations in Figure \ref{FigConfigurationsUsualResolution}, their inverses, rotations and mirror images.}
\label{FigKonfigurationer}
\end{figure}

\begin{lem}\label{LemLeftoverconfiguration}
  The configuration in Figure \ref{FigLeftoverconfiguration} cannot
  occur. 
\end{lem}

\begin{figure}
  \includegraphics[height=4cm]{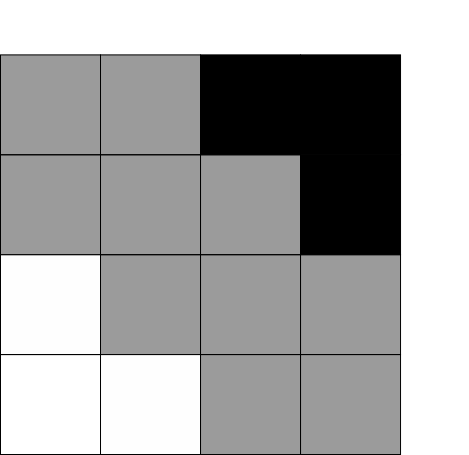}
  \caption{We show that this configuration cannot
    occur in the digitisation of an $r$-regular object.}
  \label{FigLeftoverconfiguration}
\end{figure}

\begin{figure}
\centering
  \includegraphics[scale=0.3]{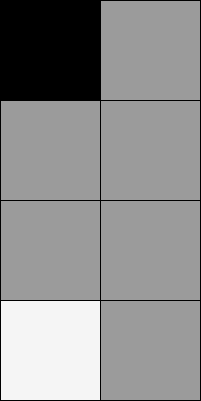}
  \hspace{1cm}
  \includegraphics[scale=0.3]{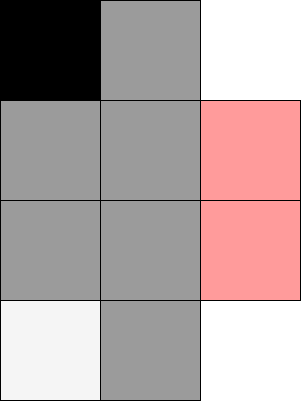}
  \caption{The right configuration does not occur in the
    digitisation. If it did, then one of the red pixels (right part of
    the figure) would have to be non-grey by Theorem
    \ref{Lem2x3grey}.}
  \label{Fig2x4Configuration}
\end{figure}

\begin{lem}\label{Lem2x4Configuration}
  The left configuration in Figure \ref{Fig2x4Configuration}
  cannot occur.
\end{lem}

\begin{lem}\label{Lem4Intersections2x2}
The boundary $\del X$ cannot intersect all four boundary edges of a configuration of $2\times 2$ grey pixels.
\end{lem}

\begin{thm}\label{Thm2x2}
The only possible configurations of $4\times 4$-configurations with $2\times 2$ grey pixels in the middle are the ones shown in Figure \ref{FigKonfigurationer13}.
\end{thm}

  \begin{figure}
    \begin{minipage}[b]{\linewidth}
      \centering
      \includegraphics[width=0.9\linewidth]{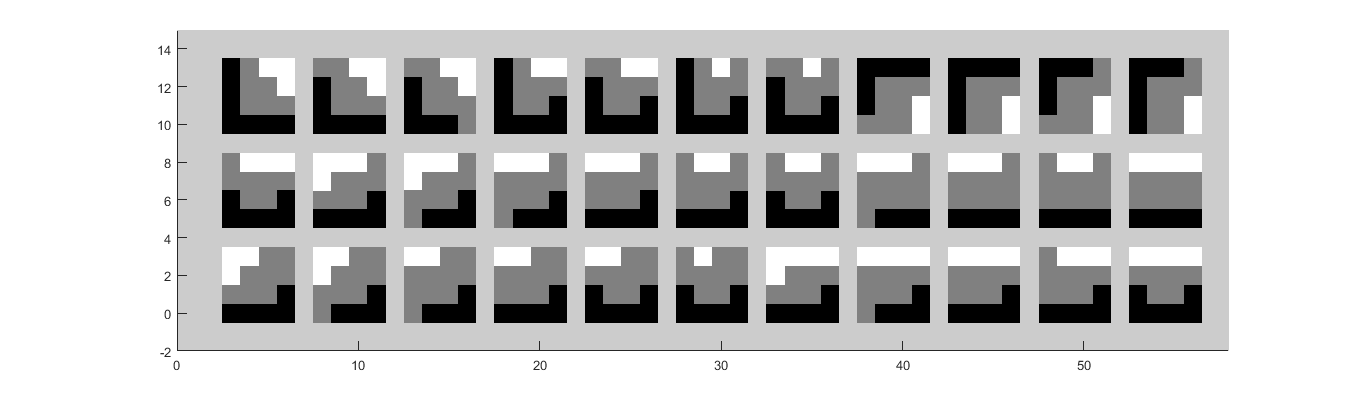}
      \caption{The 33 possible configurations of $4\times4$ pixels
        with the middle ones grey, up to rotation, reflection and
        switching of black and white pixels.}
      \label{FigKonfigurationer13}
    \end{minipage}
  \end{figure}
  
  Note that the converse is not true: There are configurations in Figure \ref{FigKonfigurationer13} that does not occur in any image of an $r$-regular object by a lattice $d\Z^2$ with $d\sqrt{2}<r$. But since the proofs of this is rather technical and the result is not relevant to our further progress, we will not discuss them here.

\section{Reconstruction of the boundary of the set}
All the work done in the previous section was leading up to the development of a reconstruction algorithm, which we will introduce in this section. The idea will be to use circle arcs to approximate the boundary of the edge. The reconstructed set will not in general be $r$-regular.

Before we start, we will introduce some points, called auxiliary points, that our reconstructed boundary must pass through. These are defined differently for different grey pixels. Hence we define

\begin{defn}
A grey pixel sitting in a $2\times 2$ configuration of grey pixels is called complex. A grey pixel that is not sitting in a $2\times 2$ configuration of grey pixels is called simple. 
\end{defn}

%
%

%

We now introduce the auxiliary points needed for the reconstruction. 


Consider a pixel edge shared by two grey pixels $A$ and $B$. If the two grey pixels sit in the same $2\times 2$ configuration of grey pixels, we introduce an auxiliary point at the midpoint of this configuration. If they do not, we introduce a point on the midpoint of their common edge, see Figure \ref{FigAuxiliaryPoints}. (Note that the two grey pixels may be part of two different $2\times 2$ configurations of grey pixels. In that case, we introduce two auxiliary points, on at the centre of each of the two $2\times2$ configurations).

\begin{figure}
\includegraphics[scale=0.5]{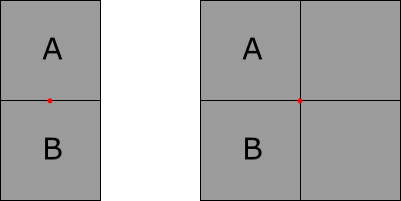}
\caption{If two grey pixels share an edge and are not a part of the same $2\times 2$ configuration of grey pixels as in the left figure, we introduce an auxiliary point (red) at the midpoint of their common edge. If on the other hand the two pixels are a part of the same $2\times2$ configuration of grey pixels, we introduce an auxiliary point at the centre of this $2\times2$ configuration.}
\label{FigAuxiliaryPoints}
\end{figure}

\begin{lem}
All simple grey pixels have between one and three auxiliary points on their boundary. All complex grey pixels have either one or two grey auxiliary points on their boundary.
\end{lem}		


\begin{proof}
For the simple pixels, look at all possible configurations in Figure \ref{FigConfigurationsUsualResolution}. For the complex pixels, look at all possible configurations in Figure \ref{FigKonfigurationer13}.
\end{proof}

\begin{lem}\label{LemRemovedPoints}
A simple pixel with just one auxiliary point on its boundary must share this point with a simple pixel with three auxiliary points on its boundary. On the other hand, a simple pixel with three auxiliary points on its boundary must share exactly one of these points with a simple pixel with just one auxiliary point on its boundary.
\end{lem}


\begin{proof}
Check all cases as presented in Figure \ref{FigConfigurationsUsualResolution}.
%
%
\end{proof}

We will now remove the auxiliary point of all simple pixels that has only one auxiliary point. By the above lemmas, this now means that all simple pixels have zero or two auxiliary points on their boundary, and all complex pixels have one or two auxiliary points on their boundary.

\begin{lem}\label{LemAuxiliaryPoints2x2}
In each $2\times 2$ configuration of grey pixels there are exactly 3 auxiliary points - one at the centre and two on the configuration boundary.
\end{lem}

\begin{proof}
Each $2\times 2$-configuration of grey pixels sits in one of the configurations in Figure \ref{FigKonfigurationer13} (up to rotation, mirroring and switching of colours). Hence we get the above theorem by checking all possible cases.
\end{proof}

\begin{thm}
For each auxiliary point $p$, there are exactly two auxiliary points with the property that there is a pixel having both $p$ and that auxiliary point on its boundary.
\end{thm}

\begin{proof}
Consider an auxiliary point $p$, sitting on the boundary of pixel $C$. If $p$ is the centre of some $2\times 2$ configuration of black pixels, then by Lemma \ref{LemAuxiliaryPoints2x2} there are only two auxiliary points on the boundary of this configuration as claimed. 

If $p$ is the midpoint of some pixel edge, consider a pixel $C$ having this edge. It is either simple (and hence has two auxiliary points on its boundary), or it is complex and hence has an auxiliary point in a corner of $C$. In both cases, there is exactly one other auxiliary point on the boundary of each of the pixels having $p$ on their boundaries.
\end{proof}

The above theorem means that there is a natural way of defining 'neighbouring auxiliary points': Two auxiliary points are neighbours if they sit on the boundary of the same pixel. 

The next step is to approximate the boundary of $X$ with curve segments: Consider an auxiliary point and its two neighbour auxiliary points. We approximate $\del X$ by circle arc segments through these three points (or, if the points are collinear, by line segments). Hence there are two curve segments through each two neighbouring auxiliary points sitting on the boundary of the same pixel $C$, one starting in one of the points, the other ending in the other point.
Each of these curve segments are graphs over the straight line $L$ through the points, so we may write them as $\gamma_1:[0,\abs{L}]\to \R$ and $\gamma_2:[0,\abs{L}]\to \R$. Then choosing a bump function $\phi:[0,\abs{L}]\to [0,1]$ with $\phi(0)=1$ and $\phi(\abs{L})=0$, we may patch a connected curve $\gamma$ together by putting $\gamma_C(t)=(1-\phi(t))\gamma_1(t)+\phi(t)\gamma_2(t)$ in each pixel $C$, see Figure \ref{FigGamma}. The resulting curve $\gamma_C$ is then also a graph over $L$, and the curve $\gamma$ is a smooth embedded submanifold of $\R^2$.
\begin{figure}
\includegraphics[scale=0.7]{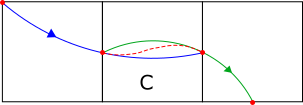}
\caption{The curve $\gamma_1$ (blue) and the curve $\gamma_2$ (green) are patched together inside $C$ using a bump function. This produces the curve $\gamma_C$ (dashed red curve).}
\label{FigGamma}
\end{figure}

\begin{lem}\label{LemGammaInConvexHull}
The path $\gamma_C$ is contained in the area bounded by $\gamma_1$ and $\gamma_2$.
\end{lem}

\begin{proof}
Since $\gamma_C$ is a graph over a line $L$ and it is a convex combination of a point on $\gamma_1$ and a point on $\gamma_2$, the curve $\gamma_C$ must lie between these two points, and hence also between the curves $\gamma_1$ and $\gamma_2$. 
\end{proof}

\begin{figure}
\includegraphics[scale=0.5]{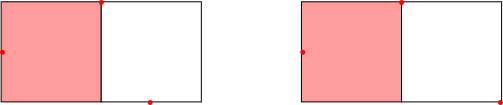}
\caption{Auxiliary points sitting in one of these configurations around the red pixel are exemptions to Lemma \ref{LemGammaInC}.}
\label{FigExceptions}
\end{figure}

\begin{lem}\label{LemGammaInC}
Consider two neighbour auxiliary points on the boundary of some pixel $C$ not sitting in a configuration like the ones in \ref{FigExceptions}. 
\begin{itemize}
\item If the two neighbour auxiliary points do not sit at the endpoints of some edge of $C$, then the curve $\gamma_C$ is contained in $C$.
\item If the two neighbour auxiliary points do sit at the endpoints of the edge shared by $C$ and some other pixel $C'$, then the curve $\gamma_C$ is contained in $C\cup C'$.
\end{itemize}
\end{lem}

\begin{figure}
\includegraphics[width=\linewidth]{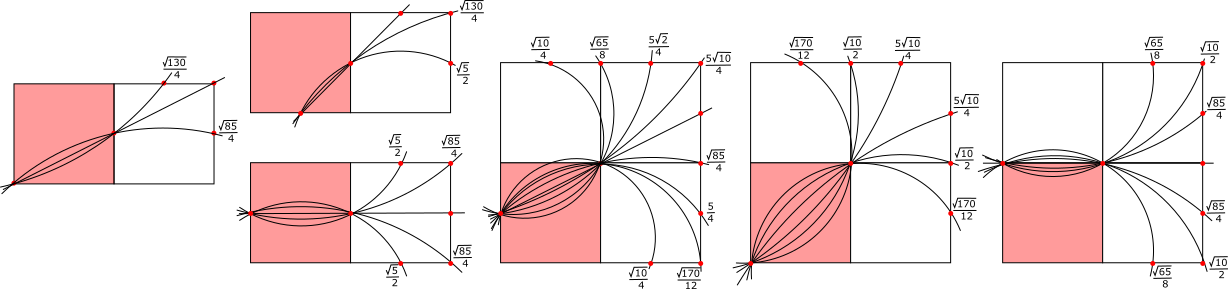}
\caption{The figure shows all possible positions of two auxiliary points on the boundary of some pixel $C$ (the red ones in the figure), and all circle arcs through these two and a third auxiliary point. The radii of the circle arcs are also calculated.}
\label{FigCircleArcs}
\end{figure}

\begin{proof}
We start by proving the lemma for the arc segments $\gamma_1$ and $\gamma_2$: We can consider all possible positions of three neighbour auxiliary points, see Figure \ref{FigCircleArcs}. Look at auxiliary points $p_1$, $p_2$ sitting in configurations around a red pixel like any of the five figures to the left. A calculation (or a look at the figures!) then show that with the exception of auxiliary points in configurations as the ones shown in Figure \ref{FigExceptions}, all possible circle arcs through two auxiliary points sitting on a red pixel are contained in that red pixel.

Now, lets argue that $\gamma_C$ is contained in $C$: By Lemma \ref{LemGammaInConvexHull}, $\gamma_C$ is contained in the area $A_C$ bounded by $\gamma_1$ and $\gamma_2$. Since $\gamma_1$ and $\gamma_2$ are both contained in $C$ which is convex, $A_C$ is also contained in $C$. Hence $\gamma_C$ must also be contained in $C$. This proves $i)$. Part $ii)$ is proved in the same way, but now looking at the right figure in Figure \ref{FigCircleArcs}. 
\end{proof}

\begin{prop}
The curves $\gamma_C$ have no self-intersections, and do not intersect each other. Hence each component of $\gamma$ is a simple closed curve.
\end{prop}

\begin{proof}
Since each segment $\gamma_C$ is a graph over some straight line $L$, we only need to show that two segments $\gamma_C$ and $\gamma_{C'}$ do not intersect. By Lemma \ref{LemGammaInConvexHull} it suffices to show that the area $A_C$ bounded by two arcs $\gamma_1$ and $\gamma_2$ in a pixel $C$ does not intersect the area $A_{C'}$ bounded by two arcs $\gamma'_1$ and $\gamma'_2$ in another pixel $C'$. 

Consider the possible circle arcs shown in Figure \ref{FigCircleArcs}. With the exception auxiliary points sitting in configurations like the ones in Figure \ref{FigExceptions}, all arc segments stays inside the pixel(s) containing both of the auxiliary points they join. Since no pixel can have more than two auxiliary points on their boundary, the only possible way that two curve segments can intersect is if one of them, say $\gamma_C$, is made using a curve $\gamma_1$ connecting points in a configuration like the one in Figure \ref{FigExceptions}. But going through all possible configurations where such a $\gamma_1$ could occur one can conclude that $\gamma$ cannot intersect itself in this case either.

That each segment of $\gamma$ is a simple closed curve follows from the fact that the segments $\gamma_C$ always connect two neighbour auxiliary points, and all auxiliary points have two neighbours. If a component of $\gamma$ were not a closed curve, it would have an endpoint (since all components of $\gamma$ are bounded) - but $\gamma$ is the join of curve segments between neighbour auxiliary points, so such an endpoint can only occur in one of the auxiliary points. But since all auxiliary points have two neighbour auxiliary points, this is impossible. 
\end{proof}

\begin{thm}\label{ThmGammaProperties}
For each component of $\del X$, there is exactly one component of $\gamma$. Each component of $\gamma$ separates the boundary components of a connected component $A$ of the set of grey pixels.
\end{thm}

\begin{proof}
Let $\del X'$ be a component of $\del X$, and let $A$ be the set of grey pixels containing points of $\del X'$. Note that $A$ cannot have any grey neighbour pixel $B$, because this would imply that $B$ contained a point from another component of $\del X$, which would mean that there were two points on different components of $\del X$ closer than $2d\sqrt{2}$ - a contradiction by Corollary \ref{CorPathInSpindle} applied to $\pi$. 
Hence $A$ is a connected component of the set of grey pixels.

Consider any chain of grey pixels in $A$, where each pixel in the chain is a neighbour of the previous and the next pixel in the chain, and each pixel appear in the chain no more than once. Assume that the start pixel and end pixel of the chain has at least two grey d-neighbours. We aim to show that the first and last pixel in such a chain is connected by a segment of $\gamma$. 

By construction, each pixel in such a chain has at least two grey d-neighbours, hence has at least one auxiliary point on its boundary. If a pixel $C$ in the chain has only one auxiliary point on its boundary, its two grey d-neighbours must sit in a $2\times2$ configuration with $C$, and hence one of its grey d-neighbours must have two auxiliary points on its boundary. Hence if we replace $C$ by its d-neighbours with two auxiliary points on their boundary, we still get a chain of pixels in $A$. Repeating this, we end up with a chain of pixels where all pixels in the chain have two auxiliary points on their boundary. The construction of $\gamma$ then yields a segment of $\gamma$ through this pixel chain.

Now by $r$-regularity of $X$, $A$ must be larger than 2 pixels and hence have at least one pixel with two grey d-neighbours. Hence $A$ contains at least one component of $\gamma$. If $A$ contained two components $\gamma'$, $\gamma''$ of $\gamma$, we could pick a chain of grey pixels connecting a pixel containing a point of $\gamma'$ with a pixel containing a point of $\gamma''$. Then by the above, the auxiliary points on the first pixel would be connected to the auxiliary points on the last pixel by a segment of $\gamma$. But then $\gamma'$ and $\gamma''$ would be connected - a contradiction. Thus for any component of $\del X$, there is exactly one component of $\gamma$.

For the second part of the statement, consider a chain of pixels following a boundary component $\del A'$ of $A$. By the above, this chain yields a segment of $\gamma'$ which is a closed curve containing $\del A'$, but not containing any other components of the boundary of $A$. Hence $\gamma'$ separates any component of $\del A$ from the others - in particular, there can be at most two boundary components of $A$, and $\gamma$ separates them. In fact, there are always two components of $\del A$: Any point $x\in\del X'$ has a black and a white $\sqrt{2}d$-ball osculating at $x$, and these balls contain the pixels in which they are centred. Since $\del X'$ separates the two balls and hence the two pixels where they are centred, so does $A$. But then $A$ must have two different boundary components, and both $\del X'$ and $\gamma'$ separates these two components.
\end{proof}

Since the set of grey pixels separates the white pixels from the black, the above theorem actually implies that $\gamma$ also separates the white pixels from the black (in the sense that any curve from a black to a white pixel must intersect $\gamma$).
We may conclude (via the Jordan Curve Theorem) that each component of $\gamma$ separates $\R^2$ into two sets, a bounded and an unbounded. From now on, $\gamma$ is the boundary of the reconstructed set, which we define as follows:

\begin{defn}
We define \emph{the reconstructed set} $\Gamma$ to be the bounded set having $\gamma$ as boundary.
\end{defn}

\section{Hausdorff distance between the boundaries of the original set and the reconstruction}

We are now ready to look at the Hausdorff distance between our reconstruction and the original object. Let us start by proving a lemma:

\begin{figure}
\includegraphics[scale=0.5]{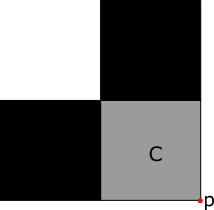}
\caption{A point $q\in\del X$ that belongs to a pixel with two black neighbour pixels sharing an edge must belong to the ball $B_d(p)$.}
\label{FigCornerPixel}
\end{figure}

\begin{lem}\label{LemCornerPixel}
Consider a grey pixel $C$ as the one in Figure \ref{FigCornerPixel}, with two black (or two white) neighbour pixels sharing a vertex. Let $p$ be the vertex of $C$ that does not belong to the two black (or white) neighbour pixels. If $q\in\del X\cap C$, then $q\in B_d(p)$.
\end{lem}

\begin{proof}
Let $q$ be as above. Then $q$ belongs to a component of $\del X$ that must enter and leave $C$ in two places, say in points $x_1$ and $x_2$. These points belong to one of the edges of $C$ having $p$ as a corner. Let $L$ be the line segment between $x_1$ and $x_2$. 

Now $q$ must belong to $\pi(L)$, which in turn must belong to $S(L,r)$ by Corollary \ref{CorPathInSpindle}, which again is contained in any ball of radius less than $r$ containing $L$ by Lemma \ref{LemSpindleIsIntersection}. The ball $B_d(p)$ contains both of the edges of $C$ that have $p$ as an endpoint, hence it also contains $x_1$ and $x_2$ and consequently $L$, $S(L,r)$ and $q$.
\end{proof}

Now that we have a suggestion for the boundary of the original set, we aim to show how good this approximation is. The first step will be to prove the following:

\begin{thm}\label{ThmHausdorffA}
Any point of $\del X$ has distance at most $d$ to the curve $\gamma$ consisting of curve segments $\gamma_C$. Hence $\sup _{y\in \del X}\inf _{x\in \gamma} d(x,y)\leq d$.
\end{thm}

This theorem, however, requires some additional lemmas:

\begin{lem}\label{LemGammaPassing2x3Grey}
If two neighbour auxiliary points sit on the common boundary edge of two grey pixels $C_1$ and $C_2$, then the curve $\gamma_{C_1}=\gamma_{C_2}$ is contained in the set $C'$ of points in $C_1\cup C_2$ that are at a distance $0.133d$ from the common edge of $C$.
\end{lem}

\begin{proof}
If two auxiliary points sit at the common boundary edge $e$ of $C_1$ and $C_2$, they must sit on the ends of $e$, i.e. they are the two common vertices of $C_1$ and $C_2$. 


By Lemma \ref{LemGammaInC}, part ii), $\gamma_{C_1}=\gamma_{C_2}$ belongs to $C_1\cup C_2$. Let $\gamma_1$ and $\gamma_2$ be the two arc segments whose merge is $\gamma_{C_1}$. Then they are both circle arcs of radius no smaller than $s=\frac{\sqrt{65}}{8}$ (see Figure \ref{FigCircleArcs}), hence they are contained in the spindle $S(e,s)$ whose height is $(s-\sqrt{s^2-\frac{1}{4}})d\approx0.133d$. Thus no point of $\gamma_1$ or $\gamma_2$ is further away than $0.133d$ from $e$. Since the curve $\gamma_{C_1}$ belongs to he area bounded by $\gamma_1$ and $\gamma_2$ by Lemma \ref{LemGammaInConvexHull}, we must also have that $\gamma_{C_1}$ belongs to $C'$.
\end{proof}

\begin{lem}\label{LemGammaPassingThroughBall}
If two auxiliary points sit on two edges of a pixel $C$ sharing a corner $p$, then $\gamma_C$ is contained in $B_d(p)\cap C$.
\end{lem}

\begin{proof}
We already know from Lemma \ref{LemGammaInC} that $\gamma_C$ belongs to $C$. Hence we only need to show that $\gamma_C$ belongs to $B_d(p)$. By Lemma \ref{LemGammaInConvexHull} it suffices to show that the area $A$ between the two curves $\gamma_1$ and $\gamma_2$ belongs to $B_d(p)$ and, since $B_d(p)$ is convex, it is even enough to show that $\gamma_1$ and $\gamma_2$ both belong to $B_d(p)$. This can be done by a calculation for all possible cases, or by looking at Figure \ref{FigGammaPassingThroughC}.
\end{proof}

\begin{lem}\label{LemGammaPassingCWithPointsOnOppositeEdges}
If two neighbour auxiliary points on the boundary of some pixel $C$ do not sit on the same edge of $C$ and neither in a configuration as the ones in Figure \ref{FigExceptions}, then the distance between any point of $\del X$ in $C$ and the curve $\gamma_C$ is less than ${d}$.
\end{lem}

\begin{proof}
Consider two auxiliary points on the boundary of pixel $C$. Suppose they sit on two opposite edges $e_1$, $e_2$ of $C$. If furthermore the two points do not sit in one of the configurations of Figure \ref{FigExceptions}, then by Lemma \ref{LemGammaInC}, the curve $\gamma_C$ belongs to $C$. Hence the curve $\gamma_C$ must run from one side of $C$ to the other without leaving $C$, see Figure \ref{FigGammaPassingThroughC} left. Thus given any point in $C$, projecting it to $\gamma_C$ along a line parallel to $e_1$ moves it no further than a distance $d$. Hence all points of $C$ is closer than $d$ to $\gamma_C$.

\begin{figure}
\includegraphics[scale=0.7]{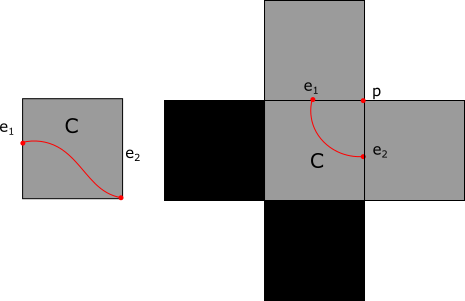}
\caption{Two different ways that a curve $\gamma$ can pass through a pixel $C$ with two auxiliary points on its boundary.}
\label{FigGammaPassingThroughC}
\end{figure}

On the other hand, suppose the two auxiliary points on the boundary of $C$ sit on the midpoint of two edges $e_1$, $e_2$ sharing a vertex $p$, see Figure \ref{FigGammaPassingThroughC} right. Then $C$ is a simple pixel, and since its auxiliary points do not sit on opposite edges, it cannot be one of the simple pixels that we removed auxiliary points from (by the proof of Lemma \ref{LemRemovedPoints}). So $C$ must have two grey d-neighbour pixels sharing the vertex $p$, and two non-grey d-neighbour pixels sharing the vertex opposite of $p$, as in Figure \ref{FigGammaPassingThroughC} right. Let us assume these two non-grey pixels to be black.

Consider at point $q\in C\cap\del X$. By Lemma \ref{LemCornerPixel}, $q$ must belong to $B_d(p)$.
Then since the path $\gamma_C$ is also contained in this ball by Lemma \ref{LemGammaPassingThroughBall} and runs from $e_1$ to $e_2$, we hit $\gamma_C$ somewhere if we move a point in $B_d(p)\cap C$ along a radius of $B_d(p)$. Such a movement displaces the point a distance of at most $d-\frac{1}{2\sqrt{2}}d$, since this is the maximal distance between a point on $\gamma_C$ and a point on $\del B_d(p)$ on the same radius of $B_d(p)$. Hence a point of $\del X\cap C$ is at most a distance $d$ from $\gamma_C$.
\end{proof}

\begin{proof}[Proof of Theorem \ref{ThmHausdorffA}]
By Lemma \ref{LemGammaPassingCWithPointsOnOppositeEdges} the theorem holds for any point of $\del X$ contained in a pixel with two auxiliary points not sitting on the same edge, and not sitting in one of the configurations of Figure \ref{FigExceptions}. Hence we need to show the result for points on $\del X$ contained in i) grey pixels with two auxiliary points sitting on the same edge, ii) the special cases in Figure \ref{FigExceptions}, iii) grey pixels with one auxiliary point on their boundary and iv) grey pixels with zero auxiliary points on their boundaries.

\emph{Ad i)}: 
By Lemma \ref{LemGammaPassing2x3Grey}, $\gamma_{C_1}$ must belong to the set $C'$ of points in $C_1\cup C_2$ closer than $0.133d$ to $e$, and by Lemma \ref{Lem2x3grey}, all points of $\del X$ in $C_1\cup C_2$ must be closer than $(\sqrt{2}-1)d$ to $e$. Since the curve $\gamma_{C_1}$  runs from one side of $C'$ to the other, then pushing a point $p\in\del X\cap C'$ orthogonally to $e$ inside $C_1\cup C_2$, we must hit $\gamma_{C_1}$ at some point. The displacement made in this manner can be no larger than $(\sqrt{2}-1+0.133)d\approx 0.55d$, hence any point of $(C_1\cup C_2)\cap\del X$ is closer than $0.55d$ to $\gamma_{C_1}=\gamma_{C_2}$.

\begin{figure}
\includegraphics[scale=0.7]{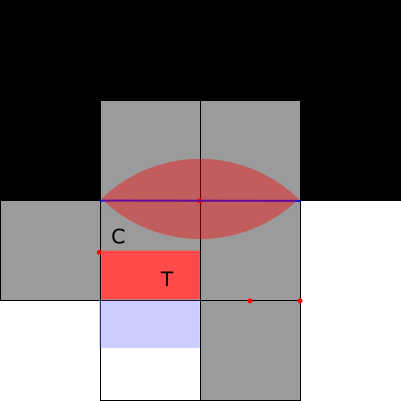}
\caption{Both of the configurations excepted from Lemma \ref{LemGammaInC} must sit in a configuration like the one shown above. We aim to show that the red rectangle cannot contain any points of $\del X$.}
\label{FigExceptedConfigurationsTriangle}
\end{figure}

\emph{Ad ii): } Now consider instead either of the cases from Figure \ref{FigExceptions}. Such a configuration must necessarily sit in a configuration like in Figure \ref{FigExceptedConfigurationsTriangle}, by looking at the possible configurations involving $2\times 2$ grey pixels in Figure \ref{FigKonfigurationer13}. We will aim to show that the rectangle $T$ in the figure, which shares two vertices with pixel $C$ and has the other two vertices at the midpoints of the vertical pixel edges of $C$, does not contain any points of $\del X$.

Look at the blue line $L$ separating the two upper grey pixels from the lower. Since there are grey pixels on both sides of $L$, $\del X$ must pass it somewhere, and since both endpoints of $L$ are black, $\del X$ must pass $L$ at least twice. Then there must be some point $p$ in one of the two upper pixels where $\del X$ has horisontal tangent, and hence the centres of the black and white $\sqrt{2}d$-balls meeting at this point sit on the vertical line through $p$. Since the pixels above the $2\times 2$ grey are black in the figure, the upper ball osculating $\del X$ at $p$ must be black, and the lower must be white.

By Corollary \ref{CorPathInSpindle} applied to $\pi$ the part of $\del X$ between two points in $\del X\cap L$ must be contained in the spindle $S(L,r)$ (shown in the figure), which contains points no further from $L$ than $(\sqrt{2}-1)d$. So $p$ cannot be further above $L$ than $(\sqrt{2}-1)d$.

Now if $p$ belonged to the right upper grey pixel, the centre of the white ball osculating $\del X$ at $p$ would belong to a grey pixel and hence colour that pixel white. So this is not possible. Therefore $p$ must belong to the upper left pixel, and be no further from $L$ then $(\sqrt{2}-1)d$. The centre of the white $\sqrt{2}d$-ball osculating $\del X$ at $p$ must therefore lie in the white pixel, no further than $\sqrt{2}d$ from $p$ and hence no further from the common edge between the white pixel and $C$ than $(\sqrt{2}-1)d$ (so somewhere in the light blue part of this white pixel). But any $\sqrt{2}d$-ball centred centred in the top half of the white pixel must contain $T$, the bottom half of pixel $C$, since $T$ and the top half of $C$ form a square with side length $d$
 Hence the white ball osculating $\del X$ at $p$ must contain all of $T$, so $T$ cannot contain any points of $\del X$.

A calculation shows that no point of $\gamma_C$ lies further above $L$ than $0.041d$. 
Hence if we take any point $q$ in $C\backslash T$ and push it along a vertical line to $\gamma_C$, we can do this without moving $q$ more than a distance $0.541d$ away. So any point of $\del X\cap C$ is closer than $d$ to $\gamma_C$.

\begin{figure}
\begin{minipage}{0.45\linewidth}
\includegraphics[scale=0.7]{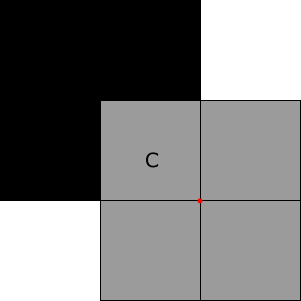}
\caption{A pixel $C$ with just one auxiliary point on its boundary must sit in a configuration as the above.}
\label{FigOneAuxiliary}
\end{minipage}
\hfill
\begin{minipage}{0.45\linewidth}
\includegraphics[scale=0.7]{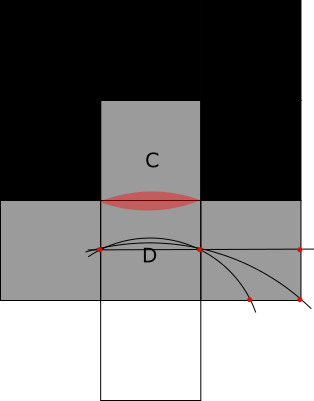}
\caption{A grey pixel with zero auxiliary points on its boundary must sit in this configuration of pixels. Then any circle arc through the auxiliary points of $D$ and one of its neighbours must look like one of the above}
\label{FigZeroAuxiliaries}
\end{minipage}
\end{figure}

\emph{Ad iii):} Consider a grey pixel $C$ with only one auxiliary point $p$ on its boundary. By construction $C$ must be complex and have two grey neighbour pixels, and two non-grey neighbour pixels, see Figure \ref{FigOneAuxiliary}. By Lemma \ref{LemCornerPixel} all points of $\del X\cap C$ must belong to the ball $B_d(p)$. Hence the distance from a point in $\del X\cap C$ to $p$ is less than $d$. Since $p$ belongs to $\gamma$, this shows the claim in this case.


\emph{Ad iv):} Consider a grey pixel $C$ without any auxiliary points on its boundary. By construction, it means that $C$ is a simple grey pixel with one grey d-neighbour pixel $D$, and three non-grey d-neighbour pixels, see Figure \ref{FigZeroAuxiliaries}.

Now, the boundary $\del X$ must pass the common edge $e$ of $C$ or $D$ in order to get into and out of $C$. Hence the part of $\del X$ that is in $C$ must be contained in $S(e,r)$.
But $S(e,r)$ contains no points in $C$ that are further from $e$ than $\sqrt{2}d-\sqrt{2d^2-\frac{d^2}{4}}<0.1d$.

Furthermore, $D$ must have one auxiliary point on the midpoint of each vertical edge - let us call these $p_1$ and $p_2$. Then an arc segment $\gamma_1$ through $p_1$ and $p_2$ and a third auxiliary point of one of the grey pixels neighbouring $C$ must lie above the straight line connecting $p_1$ and $p_2$ (just look at all possible cases, as is done in Figure \ref{FigZeroAuxiliaries}).

Hence any point $p$ in $C\cap\del X$ is closer than $0.1d$ to $e$, and any point in $e$ is closer than $\frac{d}{2}$ to $\gamma_D$. So the distance from $p$ to $\gamma$ is less than $0.6d<d$. This finishes the proof that any point of $\del X$ is closer than $d$ to $\gamma$.
\vspace{0.5cm}

For any point $x$ in $\del X$, there is a point $y'$ in $\gamma$ that is no further than a distance $d$ from $x$, meaning that
\begin{equation*}
\inf_{y\in \gamma} d(x,y)\leq d(x,y')\leq d.
\end{equation*}
Thus we get
\begin{equation*}
\sup_{y\in\del X}\inf_{x\in \gamma} d(x,y)\leq d.
\end{equation*}
\end{proof}

This proof is the first step on our way to show that $\del X$ and $\gamma$ are close to each other in Hausdorff distance. The second step is taken when we prove the following

\begin{thm}\label{ThmHausdorffB}
Any point of $\gamma$ has distance at most $d$ to the boundary $\del X$ of the original set $X$. Hence $\sup _{y\in \gamma}\inf _{x\in \del X} d(x,y)\leq d$.
\end{thm}

The proof of this theorem is very similar to the proof of Theorem \ref{ThmHausdorffA}. Again we split the proof in a couple of lemmas.

\begin{lem}\label{LemSimplePixelOppositeEdges}
Consider a simple pixel $C$ with two auxiliary points on two opposite edges. Then any point of $\gamma_C$ is closer than $d$ to some point of $\del X$.
\end{lem}

\begin{figure}
\begin{minipage}{0.55\linewidth}
\includegraphics[scale=0.4]{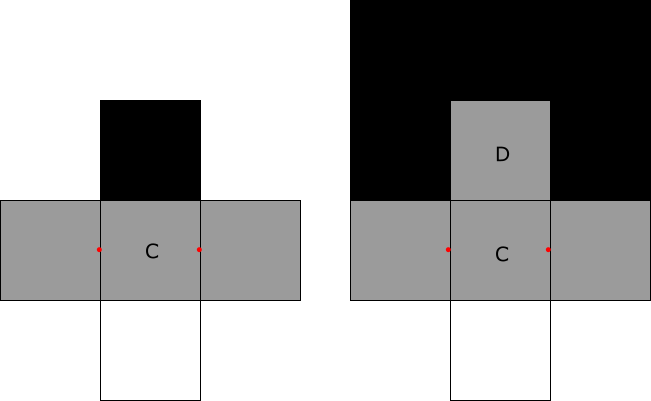}
\caption{A simple pixel $C$ with auxiliary points on opposite edges must sit in one of the configurations above}
\label{FigSimplePixelOppositeEdges}
\end{minipage}
\hfill
\begin{minipage}{0.40\linewidth}
\begin{center}
\includegraphics[scale=0.4]{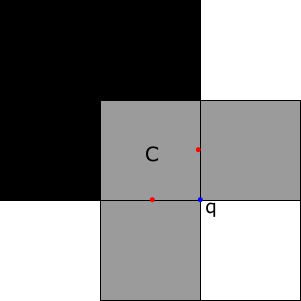}
\caption{A simple pixel $C$ with auxiliary points on vertex-adjacent edges must sit in the configuration above}
\label{FigSimplePixelVertexAdjacent}
\end{center}
\end{minipage}
\end{figure}

\begin{proof}
Notice that $C$ must sit in one of the two configurations of Figure \ref{FigSimplePixelOppositeEdges}. In the first case, pick a point $p\in\gamma_C\subseteq C$. The horisontal line in $C$ through $p$ has a black and a white endpoint, hence it must contain a point of $\del X$. Since $p$ is no further than $0.62d$ from the endpoints of this line, 
there must be a point in $\del X$ that is closer than $0.62d<d$ to $p$. 

In the second case, pick again a point $p$ in $\gamma_C\subseteq C$. Let $D$ be the pixel above $C$. Notice that $\del X$ must enter and leave $D$ by crossing $e$ in order for $D$ to be grey and both endpoints of $e$ to be black. Then by Corollary \ref{CorPathInSpindle} with $\rho$ replaced by $\pi$, any point of $\del X$ in $D$ must belong to the spindle $S(e,r)$ which contains points no further from $e$ than $\sqrt{2}d-\sqrt{2d^2-\frac{d^2}{4}}<0.1d$ by Lemma \ref{LemSpindleWidth}. So any point of $\del X\cap(C\cup D)$ must either belong to $C$ or be no further from $e$ than $0.1d$. On the other hand, a calculation 
shows that the path $\gamma_C$ is closer than $0.5d$ to $e$.

Look at a vertical line in $C\cup D$ through $p$. There must be a point on this line belonging to $\del X$, since its endpoints have different colours. Either this point belongs to $C$ (in which case they can be no further apart than $0.62d$ by the first part of the proof), or it belongs to $D$. If it belongs to $D$, it is no further from $e$ than $0.1d$, and since $p$ is no further than $0.5d$ from $e$, this point of $\del X$ must be closer than $d$ to $p$.
\end{proof}

\begin{lem}\label{LemSimplePixelVertexAdjacent}
Consider a simple pixel $C$ with two auxiliary points, located at the midpoint of two vertex-adjacent edges of $C$. Then no point of $\gamma_C$ is further away from $\del X$ than $d$.
\end{lem}

\begin{proof}
A pixel $C$ as the described must sit in a configuration as the one in Figure \ref{FigSimplePixelVertexAdjacent}. Let $q$ denote the vertex of $C$ where the two edges containing auxiliary points meet. Let $p\in \gamma_C\subseteq C$. 

Consider then line in $C$ through $p$ and $q$. Since it has endpoints of different colours, it must contain a point $s$ in $\del X$. By Lemma \ref{LemCornerPixel}, any point of $\del X\cap C$ must belong to the ball $B_d(q)$, and by Lemma \ref{LemGammaPassingThroughBall}, so must $p$. Hence $p$ and $s$ both sit on a radius of the ball $B_d(q)$. By looking at the possible curves $\gamma_C$, such two points cannot be further apart than a distance $d-\frac{1}{2\sqrt{2}}d$. So any point $p\in \gamma_C$ is closer than $d$ to $\del X$.
\end{proof}

\begin{lem}\label{LemComplexPixelSameEdge}
Consider a complex pixel $C$ with two auxiliary points located at the endpoints of some edge $e$ of $C$. Then no point of $\gamma_C$ is further away from $\del X$ than $d$.
\end{lem}

\begin{figure}
\begin{minipage}{0.45\linewidth}
\includegraphics[scale=0.4]{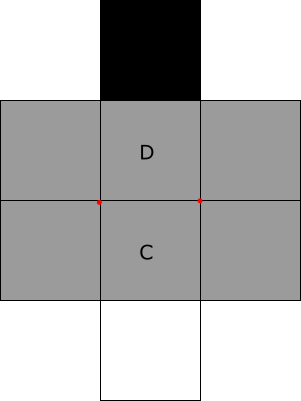}
\caption{A complex pixel $C$ with auxiliary points at the endpoints of one of its edges $e$ must sit in a configuration as the one above.}
\label{FigComplexPixelSameEdge}
\end{minipage}
\hfill
\begin{minipage}{0.45\linewidth}
\includegraphics[scale=0.4]{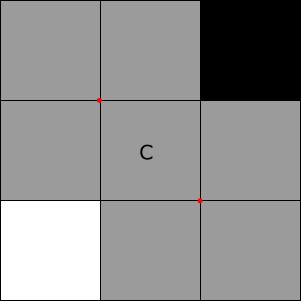}
\caption{A complex pixel $C$ with auxiliary points at two vertices of $C$ located opposite of one another must sit in a configuration as the one above.}
\label{FigComplexPixelOppositeVertices}
\end{minipage}
\end{figure}

\begin{proof}
The pixel $C$ must sit in a configuration as the one in Figure \ref{FigComplexPixelSameEdge}, by means of Lemma \ref{Lem2x3grey}. By Lemma \ref{LemGammaInC}, part ii) $\gamma_C$ must belong to $C\cup D$.

Now pick a point $p\in\gamma_C$, and look at the vertical line in $C\cup D$ through $p$. Since this line has endpoints of different colours, it must contain a point $q\in\del X$. By Lemma \ref{Lem2x3grey} again, $q$ must belong to the set of points in $C\cup D$ that are no further than $\frac{d}{2}$ from the common edge $e$ of $C$ and $D$, and by Lemma \ref{LemGammaPassing2x3Grey} $p$ is no further than $0.133d$ from $e$. Hence 
$p$ and $q$ cannot be further than $0.55d<d$ from each other.
\end{proof}

\begin{lem}\label{LemComplexPixelOppositeVertices}
Consider a complex pixel $C$ with two auxiliary points located at vertices of $C$ diagonally opposite each other. Then no point of $\gamma_C$ is further away from $\del X$ than $d$.
\end{lem}

\begin{proof}
A pixel $C$ as in this lemma must sit in a configuration as the one in Figure \ref{FigComplexPixelOppositeVertices}, by means of Lemma \ref{Thm5greys}. Let $p_1$ and $p_2$ denote the two auxiliary points on the boundary of $C$.

A calculation shows that any circle arc through $p_1, p_2$ and an auxiliary point neighbouring $p_2$ has radius greater than $d$. Hence any such circle arc is contained in the spindle $S(L,d)$ where $L$ is the line segment between $p_1$ and $p_2$, by Lemma \ref{LemSpindleIsIntersection}. By the same lemma, this means that any such circle arc $\gamma_1$ is contained in any ball of radius $d$ containing $L$. The same holds for the area bounded by the two circle arcs $\gamma_1$ and $\gamma_2$ (since $S(L,d)$ is also convex), and hence also for $\gamma_C$. Thus, if we can find some point $q$ in $\del X$ such that the $d$-ball around $q$ contains $L$ and hence $S(L,d)$ and $\gamma_C$, then any point of $\gamma_C$ must be closer than $d$ to $\del X$.

Consider the line $M$ connecting the black and white vertex of $C$. Since its endpoints have different colours, it must contain some point $q\in \del X$. Since the distance between points of $M$ and $p_1$, $p_2$ is less than $d$ everywhere, the ball $B_d(q)$ contains $p_1$ and $p_2$ and hence the spindle $S(L,d)$ between them, and we are done.
\end{proof}

\begin{lem}\label{LemComplexPixelRemaining}
Suppose $C$ is a complex pixel with an auxiliary point on one of its edge midpoints and another auxiliary point at one of the vertices of $C$. Then any point of $\gamma_C$ is closer than $d$ to $\del X$.
\end{lem}

\begin{figure}
\begin{minipage}{0.45\linewidth}
\includegraphics[scale=0.4]{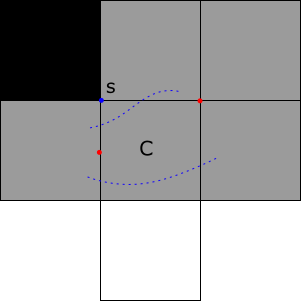}
\caption{A complex pixel as the one in Lemma \ref{LemComplexPixelRemaining} must sit in a configuration as the one above. Then $\del X$ must either intersect the upper edge of $C$ (the upper blue dashed line), or it must intersect the right vertical edge of $C$ (the lower blue line).}
\label{FigComplexPixelRemaining}
\end{minipage}
\begin{minipage}{0.45\linewidth}
\includegraphics[scale=0.4]{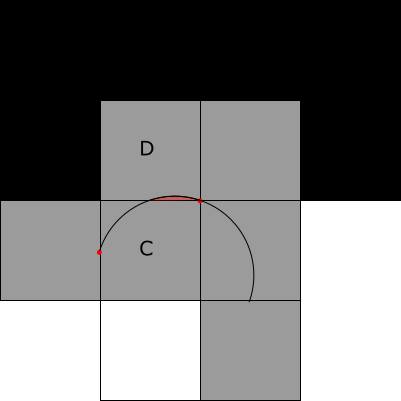}
\caption{The only case where $\gamma_C$ may not stay inside $C$ is when $C$ sit in a configuration as the one above. In this case, the part of $\gamma_C$ inside $D$ must belong to the red set.}
\label{FigComplexPixelRemainingGammaOutside}
\end{minipage}
\end{figure}

\begin{proof}
A pixel as the one in this lemma must sit in a configuration as the one shown in Figure \ref{FigComplexPixelRemaining}.

There are two cases: Either $\gamma_C$ is contained in $C$ or it is not.

Suppose first that $\gamma_C\subseteq C$. Note that $\del X$ must intersect the left vertical edge of $C$ once, since the endpoints of this edge has different colours. It cannot intersect the edge multiple times by Lemma \ref{LemTwoIntersections}. Let $q_1$ be the intersection between the left vertical edge of $C$ and $\del X$. 

Now, $\del X$ must intersect the boundary of $C$ in at least two points, one of which is $q_1$. Suppose that $\del X$ intersects $C$ somewhere on the upper edge of $C$, say in a point $q_2$. Let $L$ be the line between $q_1$ and $q_2$. Then by Corollary \ref{CorPathInSpindle} there is a path $\pi(L)$ in $\del X$ from $q_1$ to $q_2$ contained in $S(L,r)$, and by changing $q_2$ if necessary, we may assume that $\pi(L)$ does not intersect the upper edge of $C$ except at $q_2$, hence it stays inside $C$. Let $s$ be the upper left vertex of $C$.

Since $B_d(s)$ contains the left and upper edge of $C$, it contains both $q_1$, $q_2$ and hence $L$. Since $d<r$, this also means that it contains $S(L,r)$ by Lemma \ref{LemSpindleIsIntersection}, hence it contains the path $\pi(L)$. It also contains $\gamma_C$, which can be seen by considering the possible cases in Figure \ref{FigCircleArcs}.

Now, take any point $p$ on $\gamma_C$. Then it belongs to a radius of $B_d(s)$. Since $\pi(L)$ runs from one side of $C$ to another inside $C\cap B_d(s)$, there must also be a point $q$ on $\pi(L)$ lying on the same radius of $B_d(s)$ as $p$. But then $q$ and $p$ can be no further than $d$ apart, so the lemma is true in this case.

If on the other hand $\gamma_C$ belongs to $C$, but there are no points of $\del X$ on the upper edge of $C$, then there must be a point $q_2\in \del X$ on the right edge of $C$. Let $L$ be the line between $q_1$ and $q_2$. By Corollary \ref{CorPathInSpindle}, there must be a path $\pi(L)$ in $\del X$ connecting $q_1$ and $q_2$, and this path can nowhere intersect other edges of $C$.

Again pick a point $p$ in $\gamma_C$, and look at the vertical line in $C$ containing $p$. This line must be intersected by $\pi(L)$ in some point $q$, since $\pi(L)$ connects the two sides of $C$. But then $p\in \gamma_C$ and $q\in\del X$ both lie on the same line of length $d$, hence they can be no further than $d$ apart, as claimed. This concludes the proof in the case where $\gamma_C$ is contained in $C$.

Finally, assume $\gamma_C$ is not contained in $C$. Then one of the curve segments $\gamma_1$ and $\gamma_2$ are not contained in $C$ - let us say it is $\gamma_1$. Copying the results from before, we see that all points of $\gamma_C$ inside $C$ are closer than $d$ to some point of $\del X$, so it remains to show this for points of $\gamma_C$ outside $C$. Such points must lie in the set $A$ bounded by $\gamma_1$ and one of the edges of $C$ (the red set in Figure \ref{FigComplexPixelRemainingGammaOutside}).

Notice that the only case where the curve segment $\gamma_1$ is not contained in $C$ is when $C$ sit in a configuration as the one in Figure \ref{FigComplexPixelRemainingGammaOutside}. Let $D$ be the pixel above $C$ in this configuration.

The boundary $\del X$ must intersect the boundary of $D$ at least twice in order for $D$ to be grey. By Lemma \ref{LemTwoIntersections}, $\del X$ cannot intersect the right boundary of $D$ twice. Hence it must intersect the common edge $e$ of $C$ and $D$ at least once, say in a point $q$.

Now, a calculation shows for any point $q'\in e$, the ball $B_d(q')$ contains all of $A$. In particular, the ball $B_d(q)$ contains all of $A$ and hence any point of $\gamma_C$ in $D$, so any such point can be no further away from $\del X$ than $d$. This concludes the proof.
\end{proof}

\begin{proof}[Proof of Theorem \ref{ThmHausdorffB}]
The curve $\gamma$ consists of a curve segments $\gamma_C$ for each pixel $C$ with two auxiliary points on its boundary. Hence the theorem follows from the Lemmas \ref{LemSimplePixelOppositeEdges}, \ref{LemSimplePixelVertexAdjacent}, \ref{LemComplexPixelSameEdge}, \ref{LemComplexPixelOppositeVertices} and \ref{LemComplexPixelRemaining}.

Furthermore, for any point $x$ in $\gamma$, there is a point $y'$ in $\del X$ that is no further than a distance $d$ from $x$, meaning that
\begin{equation*}
\inf_{y\in \del X} d(x,y)\leq d(x,y')\leq d.
\end{equation*}
Thus we get
\begin{equation*}
\sup_{y\in\gamma}\inf_{x\in \del X} d(x,y)\leq d.
\end{equation*}
\end{proof}

\begin{cor}\label{CorHausdorff}
The reconstructed boundary $\gamma$ is closer than $d$ to the boundary of $X$.
\end{cor}

\begin{proof}
Combining Theorems \ref{ThmHausdorffA} and \ref{ThmHausdorffB}, we get that
\begin{equation*}
\sup_{y\in\gamma}\inf_{x\in \del X}d(x,y)\leq d\text{ and }\sup_{x\in \del X}\inf_{y\in\gamma} d(x,y)\leq d,
\end{equation*}
hence
\begin{equation*}
d_H(\gamma,\del X)=\text{max}\left(\sup_{y\in\gamma}\inf_{x\in \del X} d(x,y),\sup_{x\in \del X}\inf_{y\in\gamma} d(x,y))\right)\leq d.
\end{equation*}
\end{proof}

{\color{black}\section{Homeomorphism between Object and Reconstruction}
Let us now finish the proof of Theorem \ref{ThmMainResult}. We need to show that there is homeomorphism taking the reconstructed set $\Gamma$ to the original set $X$. To do so, let us start with a lemma:

\begin{lem}\label{LemHomeomorphism}
Let $M\subseteq \R^2$ be a set homeomorphic to $S^1\times [-1,1]$, and let $m\subseteq M$ be the subset homeomorphic to $S^1\times\sett{0}$. Let $\gamma: S^1\to M$ be a closed curve. Then there is a homeomorphism $f:\R^2\to\R^2$ taking $\gamma$ to $m$ fixing points in the unbounded component of $\R^2\backslash M$.
\end{lem}

\begin{proof}
Since there exists a homeomorphism of $\R^2$ taking the outer boundary component $M_o$ of $M$ to the unit circle by Schoenflies' Theorem, we may assume that $M_o$ is the unit circle. By the Annulus Theorem, the set $A$ between $M_o$ and $\gamma$ is homeomorphic to the annulus $S^1\times[\frac{1}{2},1]$ - let $g$ denote this homeomorphism. We may assume that $g$ is the identity on $M_o$ - if this is not the case, then after reversing the orientation of the map $g\vert_{M_o}$ if necessary there is an isotopy from $g(M_o)$ to $M_o$ which we may extend to an ambient isotopy of $A$ in a small tubular neighbourhood of $g(M_o)$ in $M$, and composing the result of this isotopy with $g$ we get a homeomorphism that is the identity on $M_o$.

We may continuously extend $g$ to a map $g_1$ of all of $\R^2$ by extending it by the identity on the unbounded component of $M_o^C$ (since the map $g\vert_\gamma\to g(\gamma)$ may be extended to a map of the disc bounded by $\gamma$) . Thus we get a map $g_1:\R^2\to \R^2$ taking $\gamma$ to $\frac{1}{2}S^1$ and fixing points in the unbounded component of $\R^2\backslash M_o$.

Repeating the above with $\gamma$ replaced by $m$, we also get a map $g_2:\R^2\to\R^2$ taking $S^1\times\sett{0}$ to $\frac{1}{2}S^1$. Hence the composition $g_1^{-1}\circ g_2:\R^2\to\R^2$ takes $S^1\times\sett{0}$ to $\gamma$ and fixes points in the unbounded component of $M_o$.
\end{proof}

\begin{thm}
There is a homeomorphism $H:\R^2\to\R^2$ taking $X$ to $\Gamma$. Hence $X$ and $\Gamma$ are weakly $d$-similar.
\end{thm}

\begin{proof}
Note that by Theorem \ref{ThmGammaProperties} $\gamma$ separates black pixels from white ones, and there is a 1-1 correspondence between components of $\del X$ and components of $\gamma$.

Consider an outermost component $\del X'$ of $\del X$. Since $\del X'$ is a manifold of dimension $1$, it is homeomorphic to $S^1$. Thus its tubular neighbourhood $N_{d\sqrt{2}}(\del X')$ is homeomorphic to $S^1\times [-1,1]$ (see \cite{TC}, Proposition A.10) via a map $h$ that takes the points of each normal line of length $2\sqrt{2}$ to a fiber $\sett{x}\times[-1,1]$ in $S^1\times[-1,1]$, and takes $\del X'$ to $S^1\times\sett{0}$. Moreover, $h^{-1}(S^1\times\sett{-1})$ is a subset of the set of white pixels, and $h^{-1}(S^1\times\sett{-1})$ is a subset of the black pixels. Since the boundary $\gamma$ of the reconstructed set $\Gamma$ separates black and white pixels, this means that there is a component $\gamma'$ of $\gamma$ in $N_{d\sqrt{2}}(\del X')$.

Then by Lemma \ref{LemHomeomorphism} there is a homeomorphism $f_1:\R^2\to R^2$ taking $\gamma'$ to $\del X'$ and fixing points in the unbounded component of $h^{-1}(S^1\times\sett{1})$. Since $\del X'$ and $\gamma'$ both separates black pixels from white, any component of $\del X$ inside $\del X'$ also lies inside $\gamma'$. Hence $f_1$ also takes any component of $\del X$ inside $\del X'$ to the inside of $\gamma'$.

Applying the above technique to the other components of $\del X$, we thus get a series of homeomorphisms $f_1,f_2,\dots, f_n$ that each takes one component $\del X_i$ of $\del X$ to a component $\gamma_i$ of $\gamma$. Since each homeomorphism fixes the points of the unbounded component of $\del X_i^C$, the composition $f_n\circ\dots\circ f_1:\R^2\to \R^2$ that starts by mapping the outer component(s) of $\del X$ to $\gamma$ and then works its way in, sends $\del X$ to $\gamma$. Since it also sends bounded sets to bounded sets and $\Gamma$ was the bounded set bounded by $\gamma$ and $X$ was compact, this means that $H:=f_n\circ\dots\circ f_1$ takes $X$ to $\Gamma$.

Since there is a map of $\R^2$ taking $X$ to $\Gamma$, and since $d_H(\del X,\gamma)\leq d$ by Corollary \ref{CorHausdorff}, they are weakly $d$-similar.
\end{proof}
}

\section{Example of the reconstruction algorithm}
\begin{example}
An example of this algorithm is shown in Figure \ref{FigExampleAlgorithm}. In this figure, we used the bump function
\begin{equation*}
\phi(t)=\begin{cases}
1 & \text{if $x=0$}\\
0 & \text{if $x=1$}\\
1-\frac{1}{1+exp(\frac{6}{7x}-\frac{6}{7-7x})} & \text{otherwise}
\end{cases}.
\end{equation*}
It seems from our example that the curves $\del X$ and $\gamma$ may be even closer than $d$.

\begin{figure}
\includegraphics[scale=0.8]{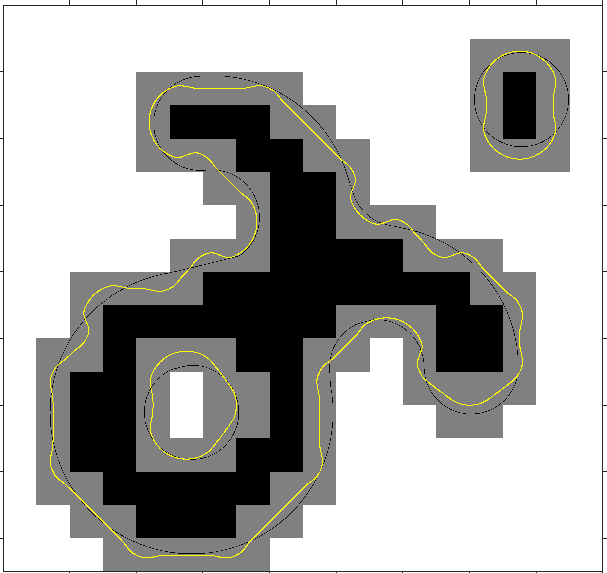}
\caption{Example of the algorithm: The thin black line is the outline of the original $r$-regular set, from which the image came. The yellow line is the reconstructed boundary of the original set.}
\label{FigExampleAlgorithm}
\end{figure}
\end{example}

%
\section{Conclusion}
In this paper we have presented restrictions on pixel configurations in digital images of $r$-regular objects at a reasonable resolution. We have used these restrictions to reconstruct the original object by constructing an object with smooth boundary that is {\color{black}weakly $d$-similar to }
the original object (where $d$ is the side length of each pixel). This tells us that our reconstruction is not far from the original object {\color{black}and has the right topology, and though that ensures that the sets are not fundamentally different}, sadly {\color{black}it is not quite enough to cover all aspects of human perception of similarity, as discussed in \cite{SK}, \cite{SLS}.}
Ongoing work is aiming at showing that the reconstructed object is in fact strongly $s$-similar to the original one for a suitable $s$. 


We do by no means believe that a Hausdorff distance of $d$ between the original object and our reconstructed object is the optimal  - in fact, we are working on obtaining even stronger bounds on their Hausdorff difference. We also believe that taking the actual intensities of each pixel into account can result in even more precise reconstruction, though there is still a lot of work to be done before we are ready to prove this.

The object that our reconstruction method outputs will in general not be $r$-regular. We may, since the boundary of the reconstructed set consists of the join of finitely many curve segments, calculate the maximal curvature of a reconstructed set - but note that a maximal curvature or $\frac{1}{s}$ is not enough to ensure that our reconstructed set is $s$-regular. Thus we have left the question of regularity of the reconstruction out of this paper, though it is also be an interesting aspect of the reconstruction process.
\newpage
\appendix
\section{Proofs of the lemmas from Section \ref{SecConfigurations}}
We here include the proofs that were omitted in Section \ref{SecConfigurations}. To lighten the notation, we will measure distances in units of $d$, so that each grid square has side length 1, and the assumption $r>d\sqrt{2}$ becomes $r>\sqrt{2}$.

\begin{proof}[Proof of Lemma \ref{LemTwoIntersections}]
  Let $x$ and $y$ be two points on the common edge of $B$ and $C$ that belongs to $\del X$, and let $L$ be the line segment
  between them. Note that since $X$ is $r$-regular and $r>\sqrt{2}$,
  $X$ is in particular $\sqrt{2}$-regular (cf. \cite{TC}, Proposition
  A.2).

  Since the distance from $x$ to $y$ is less that $\sqrt{2}$, there
  must be a path $\pi(L)$ in $\del X$ between them, where $\pi$ is the
  projection onto $\del X$, see Section \ref{SecRBasics}. Since the
  projection is continuous and fixed at the endpoints, there must be
  a point $p$ on $\pi(L)$ such that the tangent to $\del X$ at $p$ is
  horisontal. Let $p=(p_1,p_2)$. Since $p\in\del X$ and $X$ is an
  $\sqrt{2}$-regular set, there are balls
  $B_{\sqrt{2}}(x_b)\subseteq X$ and $B_{\sqrt{2}}(x_w)\subseteq X^c$
  such that
  $\overline{B_{\sqrt{2}}(x_b)}\cap\overline{B_{\sqrt{2}}(x_w)}=\sett{p}$,
  and since the tangent to $\del X$ at $p$ is horisontal, the centres
  $x_b$ and $x_w$ must lie on the vertical line through $p$.

  Note that $p\in\pi(L)\subseteq S(L,\sqrt{2})$. By Lemma
  \ref{LemSpindleWidth}, the thickness of $S(L,\sqrt{2})$ is
  $\sqrt{2}-\sqrt{2-\frac{L^2}{4}}\leq
  \sqrt{2}-\sqrt{2-\frac{1}{4}}$. So $d(p, L)\leq \sqrt{2}-\sqrt{2-\frac{1}{4}}$. Then 
  \begin{equation*}
d(x_b,L)\leq d(x_b,p)+d(p,L)\leq 2\sqrt{2}-\sqrt{2-\frac{1}{4}}<1.51
  \end{equation*}
  and
  \begin{equation*}
d(x_b,L)> d(p,L)-d(p,x_b)>\sqrt{2}-\sqrt{2-\frac{1}{4}}-\sqrt{2}=\sqrt{2-\frac{1}{4}}>1
  \end{equation*}
  So $x_b$ belongs to either $A$ or $D$ - let us say $D$. Then $D$ must be black. In fact, since the first
  and last inequality are sharp, the common edge of $B$ and $C$ must
  be interior points of $X$, and hence it contains no intersection
  points. A symmetric argument for $x_w$ shows that $x_w$ must belong to $A$, hence
  $A$ must be white, and that the common edge of $A$ and $B$ cannot
  contain any points of $\del X$.
  
  If $\del X$ is tangent to $L$ at a point $p'$, replacing $p$ with $p'$ in the above argument shows the result.
\end{proof}

\begin{proof}[Proof of Lemma \ref{LemYieldsTwoIntersections}]
Let us name the two grey pixels in the configuration $B$ and $C$, as in the right part of Figure \ref{FigYieldsTwoIntersections}. Choose boundary points $x_C\in C$ and $x_D\in D$. Both of these points are contained in a ball $B_{\sqrt{5}/2}(p)$, where $p$ is the midpoint of the common edge $e$ of pixel $B$ and $C$. By Corollary \ref{CorPathInSpindle} and Lemma \ref{LemSpindleIsIntersection} applied to the projection $\pi$ in stead of $\rho$, there is a path $\gamma$ in $\del X$ from $x_C$ to $x_D$ contained in $B_{\sqrt{5}/2}(p)$. This path must pass the line containing $e$, and since it cannot do so if passing this line means entering a black pixel, it must in fact pass the edge $e$. The endpoints of $e$ are both black, hence if $\del X$ passes $e$ once, it must also pass $e$ a second time, as $\del X$ separates black points from white ones.

But then by Lemma \ref{LemTwoIntersections} $A$ must be black and $D$ must be white (it cannot be the other way around, because then a black and a white pixel would share a corner, meaning that $\del X$ passes through that corner - which is against our assumptions).
\end{proof}


\begin{proof}[Proof of Lemma \ref{LemNo9greys}]
Let $c$ be the centre point of the grey pixel $C$. Assume $c\in X$ (the other case is similar). Then $c$ belongs to a black ball of radius $\sqrt{2}$ by Proposition \ref{PropEquivalentRregDefinitions}, hence the centre of this black ball belongs to $B_{\sqrt{2}}(c)$ and thus to either $C$ or one of its neighbours. Since the pixel containing the centre of the black ball must be entirely contained in the black ball, said pixel must be black. Hence one of the neighbours of $C$ must be black.
\end{proof}

\begin{figure}
\includegraphics[scale=0.4]{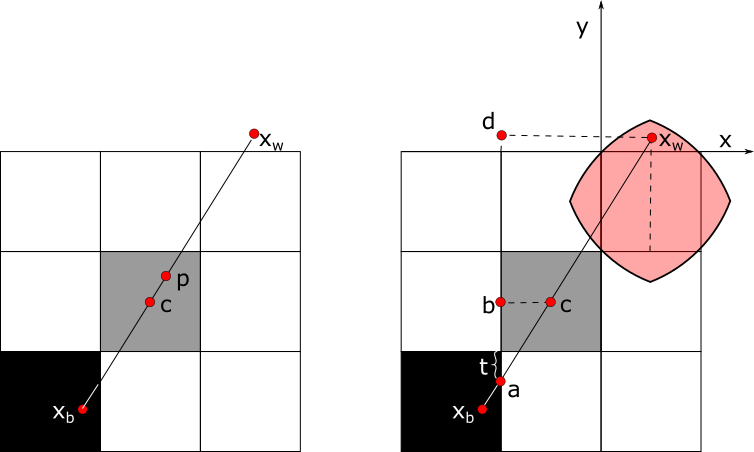}
\caption{We aim to show that the centre of the white ball tangent to $\del X$ at $p$ belongs to the red set $Y$.}
\label{Fig5greys}
\end{figure}

\begin{proof}[Proof of Lemma \ref{Thm5greys}]
Place the configuration in a coordinate system as in the figure,
  such that the pixels has side length 1. The aim will be to show that
  $x_w$ lies so close to the upper right pixel that the white ball
  $B_r(x_w)$ contains the pixel.

  Note first that the line $l$ through $x_b$ and $c$ passes an edge of
  the black pixel, say the right edge. Let $a$ be this intersection
  point, see Figure~\ref{Fig5greys}. To study the limit case, we
  will first assume that $x_b=a$.

  Since $d(a,c)<r$, then rotating the line segment from $a$ to $c$
  about $c$ with an angle $\pi$ one sees that the line segment of
  length $2r$ from $a$ through $c$ has an endpoint $x_w=(x_1,x_2)$
  with $x_1\geq 0$.

  Let $b$ be a point directly above $a$ such that
  $\angle abc=\frac{\pi}{2}$, and let $d$ be a point directly above
  $a$ such that $\angle adx_w=\frac{\pi}{2}$, see
  Figure~\ref{Fig5greysAdd}. Finally, let
  $t=d(a,b)-\frac{1}{2}\geq 0$.

  Now,
  $\tan(\angle bac)=\frac{d(b,c)}{d(a,b)}=\frac{1/2}{1/2+t}\leq 1$, so
  $\angle bac\leq\frac{\pi}{4}$. Thus we have that
  $\frac{d(d,x_w)}{2\sqrt{2}}=\sin(\angle bac)\leq
  sin\left(\frac{\pi}{4}\right)=\frac{1}{\sqrt{2}}$, so
  $x_1=d(d,x_w)-1\leq 1$. So $0\leq x_1\leq1$. Also, by rotational
  symmetry and since $d(a,c)\leq r$, we must have $x_2\geq -1+t$.

  Let $Y$ be the intersection of the four $r$-balls with centres
  $(0,0)$, $(1,0)$, $(0,-1)$ and $(1,-1)$ (the corners of the upper
  right pixel) - this is the red set in Figure~ \ref{FigCloseup5greys}. Then
  $Y$ is a convex set containing $(0,0)$, $(1,0)$, $(0,-1)$ and
  $(1,-1)$, and hence the entire upper right pixel. Any
  point in $Y$ is closer than $r$ to all the points $(0,0)$, $(1,0)$,
  $(0,-1)$ and $(1,-1)$, so an $r$-ball with centre in $Y$ contains
  all of the upper right pixel. Hence we aim to show that $x_w\in Y$.

  Notice that the two triangles $abc$ and $adx_w$ are
  equiangular. Hence
  \begin{equation*}
    \frac{1+x_1}{2\sqrt{2}}=\frac{\frac{1}{2}}{d(a,c)}=\frac{\frac{1}{2}}{\sqrt{(t+\frac{1}{2})^2+\frac{1}{4}}}\Rightarrow t=\sqrt{\frac{2}{(x_1+1)^2}-\frac{1}{4}}-\frac{1}{2}.
  \end{equation*}
  Furthermore, we have that
  \begin{equation*}
    8=(2+t+x_2)^2+(1+x_1)^2,
  \end{equation*}
  so
  \begin{equation*}
    x_2=\sqrt{8-(1+x_1)^2}-2-t=\sqrt{8-(1+x_1)^2}-\frac{3}{2}-\sqrt{\frac{2}{(x_1+1)^2}-\frac{1}{4}}
  \end{equation*}
  So we may express $x_2$ as a function of $x_1$. Since
  $0\leq x_1\leq 1$ and $x_2\geq -1$, we need only check that $x_w$
  lies under the upper border of $Y$ on the interval $[0,1]$, i.e. we
  must check that $x_2(x_1)$ lies under the function
  \begin{equation*}
    \tilde{f}(x_1)=\begin{cases}
      \sqrt{2-(x_1-1)^2}-1 & x_1\leq\frac{1}{2},
      \\
      \sqrt{2-x_1^2)}-1 & x_1\geq\frac{1}{2},
    \end{cases}
  \end{equation*}
  on $[0,1]$. However, it turns out to be easier to check that
  $x_2(x_1)$ lies under the graph of the function
  \begin{equation*}
    f(x_1)=\begin{cases}
      (\sqrt{7}-\frac{1}{2})x_1 & 0\leq x_1\leq\frac{1}{2},
      \\
      -(\sqrt{7}-\frac{1}{2})x_1+(\sqrt{7}-\frac{1}{2}) & \frac{1}{2}\leq x_1\leq1.
    \end{cases}
  \end{equation*}
  Note that $f([0,1])\subseteq Y$, since the image of $f$ is just the
  union of two line segments, both of which have endpoints in the
  convex set $Y$.

  To show that $x_2$ lies somewhere below $f$, note first that
  \begin{align*}
    \frac{d^2}{dx_1^2}x_2={} &-\sqrt{8-(x_1+1)^2}-\frac{(x_1+1)^2}{\sqrt{8-(x_1+1)^2}^3}\\
    &+\frac{4}{\sqrt{2(x_1+1)^2-\frac{1}{4}(x_1+1)^4}^3}-\frac{6}{\sqrt{2(x_1+1)^6-\frac{1}{4}(x_1+1)^8}}\\
    \leq {} & \frac{4}{\sqrt{2(x_1+1)^2-\frac{1}{4}(x_1+1)^4}^3}-\frac{6}{\sqrt{2(x_1+1)^6-\frac{1}{4}(x_1+1)^8}}\\
    = {} &\frac{1}{2(x_1+1)^3}\left( \frac{3x_1^2+6x_1-13}{\sqrt{2(x_1+1)^2-\frac{1}{4}(x_1+1)^4}^3}\right)\\
    \leq {}  & 0,
  \end{align*}
  where the first inequality comes from the fact that the first two
  terms of the derivative is negative, and the last inequality comes
  from observing that $3x_1^2+6x_1-13\leq 0$ on $[0,1]$.

  Now, we want to show that $f-x_2\geq 0$. Note that since $f$ is
  (piecewise) linear, we get that
  \begin{equation*}
    \frac{d^2}{dx_1^2}(f-x_2)=\frac{d^2}{dx_1^2}(-x_2)\geq 0
  \end{equation*}
  on $[0,1/2]$ and $[1/2,1]$, so $\frac{d}{dx_1}(f-x_2)$ is
  increasing. Now,
  \begin{equation*}
    \frac{d}{dx_1}(x_2)=-\frac{1+x_1}{\sqrt{8-(1+x_1)^2}}+\frac{2}{\sqrt{\frac{2}{(1+x_1)^2}-\frac{1}{4}}}\frac{1}{(1+x_1)^3},
  \end{equation*}
  so on $[0,1/2]$
  \begin{align*}
    \frac{d}{dx_1}(f-x_2)&=\sqrt{7}-\frac{1}{2}+\frac{1+x_1}{\sqrt{8-(1+x_1)^2}}-\frac{2}{\sqrt{\frac{2}{(1+x_1)^2}-\frac{1}{4}}}\frac{1}{(1+x_1)^3}
    \\
    &\geq \frac{d}{dx_1}(f-x_2)\vert_{x_1=0}
    \\
    &=\sqrt{7}-\frac{1}{2}+\frac{1}{\sqrt{7}}-\frac{4}{\sqrt{7}}
    \\
    &>0,
  \end{align*}
  and on $[1/2,1]$,
  \begin{align*}
    \frac{d}{dx_1}(f-x_2)&=-\sqrt{7}+\frac{1}{2}+\frac{1+x_1}{\sqrt{8-(1+x_1)^2}}-\frac{2}{\sqrt{\frac{2}{(1+x_1)^2}-\frac{1}{4}}}\frac{1}{(1+x_1)^3}\\
    &\leq \frac{d}{dx_1}(f-x_2)\vert_{x_1=1}\\
    &=-\sqrt{7}+\frac{1}{2}+{1}-\frac{1}{2}\\
    &<0.
  \end{align*}
  So $f-x_2$ is increasing on $[0,1/2]$ and decreasing on
  $[1/2,1]$. Hence, on $[0,1/2]$
  \begin{equation*}
    (f-x_2)(x_1)\geq (f-x_2)(0)=-\frac{\sqrt{7}}{2}+\frac{3}{2}>0
  \end{equation*}
  and similarly, on $[1/2,1]$
  \begin{equation*}
    (f-x_2)(x_1)\geq (f-x_2)(1)=0.
  \end{equation*}
  Putting the last two equations together we see that
  $f(x_1)-x_2(x_1)\geq 0$ everywhere on $[0,1]$, so $x_2\leq f$ as
  claimed. So if $x_b=a$, then $x_w$ belongs to the set $Y$.

  \begin{figure}
    \includegraphics[scale=0.5]{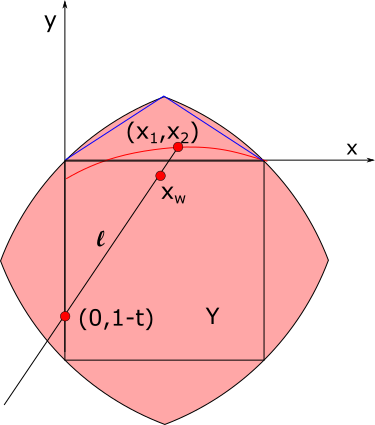}
    \caption{Close-up on the upper right pixel. The red graph is the
      graph of $x_2(x_1)$, and the blue graph is the graph of
      $f(x_1)$. The point $(x_1,x_2)$ on $l$ is chosen such that
      $d((x_1,x_2),a)=2r$, and hence $x_b$ lies closer to $c$ than
      $(x_1,x_2)$ does.}
    \label{FigCloseup5greys}
  \end{figure}

  Suppose now that $x_b$ is just any point in the lower left pixel
  that is closer than $\sqrt{2}$ to $c$, and suppose that the line $l$
  through $x_b$ and $c$ leaves the lower left pixel in a point $a$ on
  the right pixel edge.  Then $(0,1-t)$ lies on $l$ and inside $Y$. By
  what we just showed, the point $(x_1,x_2)$, $x_1\geq0$, on $l$ that
  is at a distance $2r$ from $a$ is also in $Y$, hence the entire line
  segment from $(0,1-t)$ to $(x_1,x_2)$ is in $Y$, since $Y$ was
  convex, see Figure~\ref{FigCloseup5greys}.  But noticing that
  $r\geq d(a,c)=d(c,(0,1-t))$, we get that
  \begin{equation*}
    d(x_b,(0,1-t))=d(x_b,c)+d(a,c)\leq 2r
  \end{equation*}
  and
  \begin{equation*}
    d(x_b,(x_1,x_2))=d(x_b,a)+d(a,x_w)\geq d(a,x_w)=2r.
  \end{equation*}
  Combining these equations, we see that
  $d(x_b,(0,1-t))\leq d(x_b,x_w)\leq d(x_b,(x_1,x_2))$. Hence $x_w$
  belongs to the line segment between $(0,1-t)$ and $(x_1,x_2)$ which
  was contained in $Y$, so $x_w\in Y$.
\end{proof}

\begin{proof}[Proof of Lemma \ref{LemSkalfarvessorte1}]
Let us show $(i)$. Let $x$ and $y$ be
  corner points of the two pixels as in Figure
  \ref{FigSkalfarvessorte1}, and let $L$ denote the line between them.

  If there are points of $ X^C$ in $L$, then $\del X$ must either be tangent to $X$ or intersect $L$ in several points (since the endpoints of $L$ clearly all
  belong to $X$). By Lemma \ref{LemTwoIntersections} this means that
 either the pixel above $C$ or the pixel below pixel $B$ is white. Both of these pixels share a corner with a black pixel. But by the proof of Lemma \ref{LemTwoIntersections}, the black corner point must be an interior point of $X^C$ and hence white - a contradiction. 
So $L\subseteq\Int (X)$.

  If $B$ is not black, pick white points $b\in \Int (B)$ and
  $c\in\Int (C)$. Let $L_{bc}$ denote the line between them. Then
  there is a path $\rho_{X^C}(L_{bc})$ in $X^C\cup\del X$ connecting
  $b$ and $c$, and this path belongs to all balls of radius less than
  $r$ that contains both $b$ and $c$, (cf. Section
  \ref{SecRBasics}). In particular,
  $\rho_{X^C}(L_{bc})\subseteq {B_{\sqrt{2}}(x)}$, since this ball
  contains all of $B$ and all of $C$.

  Let $\gamma$ be the piecewise linear path from $z$ though $x$ and
  $y$ to $w$. Then $\gamma$ is contained in $\Int(X)$ and separates
  $B$ from $C$ inside $B_{\sqrt{2}}(x)$. Hence $\rho_{X^C}(L_{bc})$
  must intersect $\gamma$ somewhere, but this is impossible, since
  $\gamma\subseteq\Int(X)$ and
  $\rho_{X^C}(L_{bc})\subseteq X^C\cup\del X$. So $B$ cannot contain
  any white points, and hence it must be black.

  A similar reasoning can be applied to $A$: If $A$ is not black, pick
  a white point $a\in A$, and let $L_{ac}$ denote the line between $a$
  and $c$ (where $c\in C$ is the point we chose earlier). Then there
  is a path $\rho_{X^C}(L_{ac})$ in $X^C\cup\del X$ connecting $a$ and
  $c$, and this path must belong to the ball $B_{\sqrt{2}}(x)$. But
  since $\gamma$ also separates $A$ from $C$ inside
  $B_{\sqrt{2}}(x)$, $\rho_{X^C}(L_{ac})$ must intersect $\gamma$
  somewhere. But this is impossible, since $\gamma\subseteq\Int(X)$
  and $\rho_{X^C}(L_{ac})\subseteq X^C\cup\del X$ as before. So $A$
  cannot contain any white points, and hence it must also be black.
  
  The second part of this proof also proves $(ii)$. To prove $(iii)$, we apply Lemma \ref{Thm5greys} to argue that one of the pixels $A_1$, $A_2$, $A_3$, $B_1$, $B_2$, $B_3$ is black.
  
Indeed, suppose none of the pixels $A_1$, $A_2$, $A_3$, $B_1$, $B_2$, $B_3$ were black. Then $A_1$, $A_3$, $B_1$ and $B_3$ would have to be grey, since black and white pixels cannot be neighbours by assumption. But then by Remark \ref{Rem3x3Pix}, either $A_2$ or $B_2$ would have to be black - a contradiction. So one of the pixels has to be black.
   If $A_1$, $A_3$, $B_1$ or $B_3$ is black, we are in situation $(i)$ and may use this proof to complete the proof of $(iii)$. If $A_2$ is black, and neither $A_1$ or $A_3$ is black Lemma \ref{LemYieldsTwoIntersections} shows that both $B_1$ and $B_3$ is black, and we are again in the situation of case $(i)$. If $A_2$ and either $A_1$ or $A_3$ is black, we are in the situation of case $(i)$. The proof works equivalently if $B_2$ is black. So $(iii)$ is true when one of the pixels $A_1$, $A_2$, $A_3$, $B_1$, $B_2$, $B_3$ are black.
\end{proof}

\begin{figure}
  \begin{subfigure}{0.3\linewidth}
    \centering
    \includegraphics[scale=0.4]{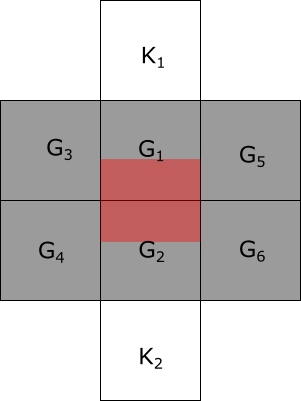}
    \caption{We are considering 6 grey pixels in a $2\times3$
      combination, and we wish to show that $\del X\cap(G_1\cup G_2)$
      belongs to the red set in the figure, and that one of the pixels
      $K_1$, $K_2$ must be black, and the other one white.}
    \label{Fig6greysAdd}
  \end{subfigure}\hspace{0.5cm}
  \begin{subfigure}{0.3\linewidth}
    \centering
    \includegraphics[scale=0.4]{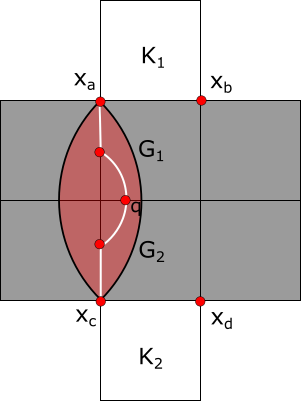}
    \caption{If both $K_1$ and $K_2$ were white, then points $x_a$
      and $x_c$ would be joined by path in $X^C\cup \del X$, and so
      would points $x_b$, $x_d$ (these are the white paths in the
      figure).}
    \label{Fig6greys2spindlesAdd}
  \end{subfigure}
  \hspace{0.5cm}
  \begin{subfigure}{0.3\linewidth}
    \includegraphics[scale=0.4]{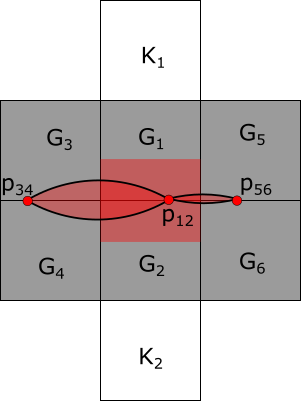}
    \caption{The projection $\pi$ yields a path $\gamma$ in
      $\del X$ from $p_{34}$ through $p_{12}$ to $p_{56}$, and this
      path lives inside the spindles between points $p_{34}$ and
      $p_{12}$, and points $p_{12}$ and $p_{56}$.}
    \label{Fig6greysWithSpindlesAdd}
  \end{subfigure}
\end{figure}

\begin{proof}[Proof of Lemma \ref{Lem2x3grey}]
  Let us start by discussing (i).

  Consider $G_1$ as in Figure \ref{Fig6greys2spindlesAdd}, and look at the configuration of $3\times3$ pixel with $G_1$ as the centre pixel. Then all but the three upper neighbours of $G_1$ are grey. By Remark \ref{Rem3x3Pix}, the upper d-neighbour $K_1$ of $G_1$ cannot be grey, hence it must be either black or white. By the same reasoning, $K_2$ must be either black or white.

  It remains to prove that $K_1$ and $K_2$ cannot have the same
  colour, so suppose that $K_1$, $K_2$ are both white. Let $x_a$,
  $x_b$ be the lower corners of $K_1$ and $x_c$, $x_d$ the upper
  corners of $K_2$, as in Figure~\ref{Fig6greys2spindlesAdd}. Note
  that these points are all elements of $X^C\cup\del X$.

  Let $L_{ac}$ be the line segment between $x_a$ and $x_c$, and let
  $L_{bd}$ be the line segment between $x_b$ and $x_d$. Since
  $d(x_a,x_c)=d(x_b,x_d)=2<2\sqrt{2}$, the map $\rho$ maps these line
  segments to continuous paths in $X^C\cup\del X$ by projecting points
  of $\text{Int}(X)$ to $\del X$ and fixing all other points.

  Now, $L_{ac}$ and $L_{bd}$ cannot both be contained entirely in
  $\overline{X^C}$, since this would imply that $\rho_{X^C}$ kept them
  fixed. But since $G_1$, $G_2$ were grey, they must contain a point
  of $\text{Int}(X)$ which in turn would belong to some black $r$-ball
  $B_r(x)\subseteq X$. However, an interior point of such an $r$-ball would
  have to intersect the boundary of $G_1\cup G_2$, which hence cannot
  be a subset of $X^C\cup\del X$.

  So assume that $\rho_{X^C}$ does not fix $L_{ac}$. Then there is a
  point $q$ on $\rho_{X^C}(L_{ac})$ that is furthest from $L_{ac}$ and hence is a point on $\del X$ with a vertical tangent. This point must
  belong to $S(L_{ac},r)$ since all of $\rho_{X^C}(L_{ac})$ does, and there
  must be a black and a white ball that are tangent to $\del X$ at
  $q$. Since the thickness of $S(L_{ac},r)\leq\sqrt{2}-1$, we must have $d(q,L_{ac})\leq\frac{1}{2}$. But then the centre $x$ of the left $\sqrt{2}$-ball tangent to $\del X$ at $q$ must satisfy 
  \begin{equation*}
  1<d(x,q)-d(q,L_{ac})\leq d(x,L_{ac})\leq d(x,q)+d(q,L_{ac})<2\sqrt{2}-1
  \end{equation*}  
   So the centre $x$ of a white or black $\sqrt{2}$-ball belongs to one of the grey pixels $G_5$, $G_6$ left of $L_{ac}$, which is
  impossible, since that pixel would thus be entirely contained in the ball, and hence not grey. So we cannot have that both $K_1$ and $K_2$ have the
  same colour, completing the proof of (i).

  For (ii), let $N$ be the line separating the upper three grey pixels
  from the lower three grey pixels.  We wish to prove that $\del X$
  intersects $N$ at least two times inside the four leftmost grey pixels,
  and also at least two times in the four rightmost grey pixels.

  Let $L_{ac}$ be the vertical line separating $G_3$ and $G_4$ from the
  other grey pixels, and similarly, let $L_{bd}$ be the vertical line
  separating $G_5$ and $G_6$ from the others. Pick a boundary point
  $x_i$ in each of the grey pixels $G_i$, $i=1,\dots, 6$, and let
  $L_{ij}$ be the line segment joining $x_i$ and $x_j$,
  $i,j=1,\dots,6$.

  Using the projection $\pi:N_r\to\del X$, we know that there is a
  path $\pi(L_{12})$ in $\del X$ from $x_1$ to $x_2$, and this path
  must necessarily cross $N$ somewhere. If it passes $L_{ac}$ on the way,
  it must do so at least twice. By an argument similar to the one we
  used to prove part (i), there must then be a boundary point $q$ with
  vertical tangent line. Since $q$ belongs to $S(L_{12},r)$, the right
  osculating ball at $q$ must belong to either $G_5$ or $G_6$, which
  yields a contradiction. Hence $\pi(L_{12})$ does not intersect
  $G_L$, and by a symmetric argument, it does not intersect $L_{bd}$
  either. So it must intersect $N$ at a point $p_{12}$ on the common
  edge of $G_1$ and $G_2$.

  A similar argument shows that the line segments $\pi(L_{34})$ must intersect $N$ in a point $p_{34}$ the common edge
  of $G_3$ and $G_4$, and hat $\pi(L_{56})$ intersects $N$ in a point $p_{56}$ the common edge of $G_5$
  and $G_6$, respectively.

  So we have three points of $\del X$ on $N$. Using the projection on the
  line segments between them, we get a path $\gamma$ in $\del X$ from
  $p_{34}$ through $p_{12}$ to $p_{56}$, and this path must live
  inside the spindles $S(\sqrt{2},p_{12},p_{34})$, and
  $S(\sqrt{2},p_{12},p_{56})$, see
  Figure~\ref{Fig6greysWithSpindlesAdd}.

  Since the maximum height of such a spindle is $\sqrt{2}-1$,
  then $\gamma$ must belong to the red part of the pixels
  $G_1\cup G_2$. Note that there cannot be any other elements of
  $\del X$ in $G_1\cup G_2$ than those of $\gamma$, for suppose $y_1$
  was a point in $G_1\cap \del X$ that did not belong to $\gamma$, and
  let $y_2$ be a point on $\gamma$ on the same vertical line as
  $y_1$. Let $L_y$ be the line segment between $y_1$ and $y_2$. Since
  $y_1$ and $y_2$ would be closer than $2r$ to each other, they would
  be connected through the path $\pi(L_y)$, and by an argument similar
  to the former ones, there would be a point $y$ on $\pi(L_y)$ where
  $X$ has horisontal normal vector. Suppose $y$ is located to the left
  of the vertical line through the centre of $G_1$. There is a black
  and a white ball that are tangent to $\del X$ at $y$, and one of
  them would have a centre $x$ at a distance $\sqrt{2}$ to the right
  of $y$. But since $y$ belonged to the left half of the
  configuration, this means that the centre $x$ would belong to one of
  the grey pixels $G_1$, $G_2$, $G_5$ or $G_6$ - which would give a
  contradiction.

  So since the path $\gamma$ belongs to the red part of $G_1\cup G_2$,
  the proof is complete.
\end{proof}

\begin{proof}[Proof of Lemma \ref{LemImpossible2x3}]
Suppose this configuration did occur, and look at the $3\times3$-configuration that also includes the three pixels to the right of it (see Figure \ref{Fig2x3Impossible}, right). Not all of the red pixels can be grey, since this would violate Lemma \ref{Thm5greys} and Lemma \ref{LemNo9greys}. Hence one of them must be another colour, say black.

If the upper red pixel were black, Lemma \ref{LemYieldsTwoIntersections} would require the bottom grey pixel to be white - a contradiction.

If the middle red pixel were black, then by the first part of Lemma \ref{LemSkalfarvessorte1} one of the grey ones in the middle column of the configuration would be so, too - a contradiction again.

If the bottom red pixel were black, then by the third part of Lemma \ref{LemSkalfarvessorte1} one of the grey ones in the middle column of the configuration would be so, too - yet another contradiction.

Hence there can be no legal way to colour the red pixels, so this configuration cannot occur.
\end{proof}

\begin{figure}
\includegraphics[scale=0.4]{Impossible3x3}\hfill
\includegraphics[scale=0.4]{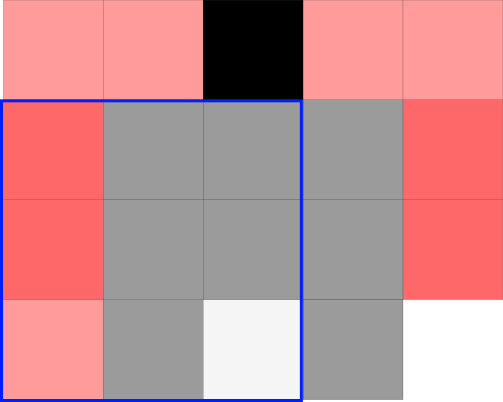}\hfill
\includegraphics[scale=0.4]{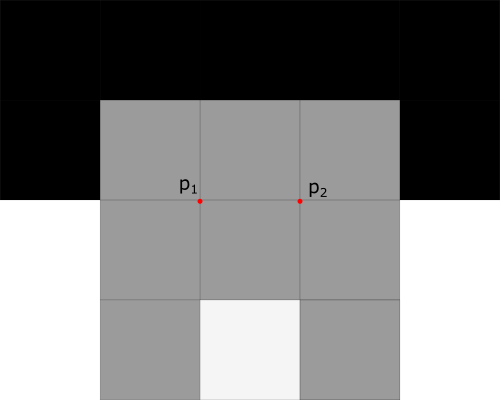}
\caption{The configuration to the left in this figure cannot occur in the image of an $r$-regular object with $d\sqrt{2}<r$. As a first part of the proof, we aim to show that the configuration is part of a larger configuration looking like the one on the right.}
\label{FigImpossible3x3}
\end{figure}

\begin{proof}[Proof of Lemma \ref{LemTrickyConfiguration}]
Before we go on to the proof, we will state and prove the following lemma for later use:
\begin{lem}\label{LemCircleIntersections}
Let $A\subseteq\R^2$ be a convex polygon, $r>0$. The intersection of all $r$-balls centred in $A$ is equal to the intersection of all $r$-balls centred at the vertices of $A$.
\end{lem}

\begin{proof}
It suffices to show the theorem for lines, since if it holds for lines, it holds for any edge of $A$, and hence for all of $A$ by convexity.

So let $A$ be a line with endpoints $(x_1,0)$, $(x_2,0)$. Let $(p_1,p_2)\in B_r((x_1,0))\cap B_r((x_2,0))$ and $(x,0)\in A$. Assume $p_1\leq x$ - the other case is symmetrical. Then $\norm{(x,0)-(p_1,p_2)}^2=(x-p_1)^2+p_2^2<(x_2-p_1)^2+p_2^2=\norm{(x_2,0)-(p_1,p_2)}^2<r^2$, because $x_2$ must be further from $p_1$ than $x$ when $p_1<x<x_2$. But this means that $B_r((x_1,0))\cap B_r((x_2,0))\subseteq B_r((x,0))$, so $B_r((x_1,0))\cap B_r((x_2,0))\cap B_r((x,0))=B_r((x_1,0))\cap B_r((x_2,0))$.
\end{proof}

Now we turn to the proof of Lemma \ref{LemTrickyConfiguration}.

Suppose the configuration did occur. Let us first argue that then it must be a part of a larger configuration looking like the one in Figure \ref{FigImpossible3x3}, right.

Since the top 6 pixels of the configuration to the left in the figure are grey and the lower middle one is white, the middle pixel just above the configuration must be black by Lemma \ref{Lem2x3grey}, as indicated in Figure \ref{FigImpossible3x3} middle.

Similarly, look at the left neighbour configuration of the one in the figure (this is the one with the blue frame in Figure \ref{FigImpossible3x3}, middle). This configuration has a grey middle and a white corner. Combining Remark \ref{Rem3x3Pix} and Lemma \ref{LemSkalfarvessorte1}, this means that either the top left or the middle left pixel (the two darker red pixels in the figure) must also be black. The same is true for the right neighbour configuration. By Lemma \ref{LemSkalfarvessorte1} $(i)$ and $(iii)$, if one of the two dark red pixels to the left is black, then the upper dark red pixels must be black, and so must the 4 red pixels in the top row in Figure \ref{FigImpossible3x3}, too. So if the configuration did occur in a digital image, it would have to sit in a configuration like the one in Figure \ref{FigImpossible3x3}, right.

Now, consider the two red points $p_1$ and $p_2$ at the top corners of the centre pixel in Figure \ref{FigImpossible3x3} right. These cannot be black: If they were, then $\del X$ must intersect one of the edges of the top left and right grey pixel at least twice, which would violate Lemma \ref{LemTwoIntersections}. So they must both be white.

Now, since all corners of the centre pixel $C$ are white and the pixel itself is grey, the boundary $\del X$ must intersect at least one edge of $C$ at least twice. It cannot be the bottom edge, and it cannot be either of the two vertical edges either by Lemma \ref{LemTwoIntersections}, so $\del X$ must intersect the line between $p_1$ and $p_2$ at least twice. We now aim to show that this line is in fact contained in the set of white points, so that it cannot contain any points of $\del X$, giving us a contradiction.

\begin{figure}
\includegraphics[scale=0.4]{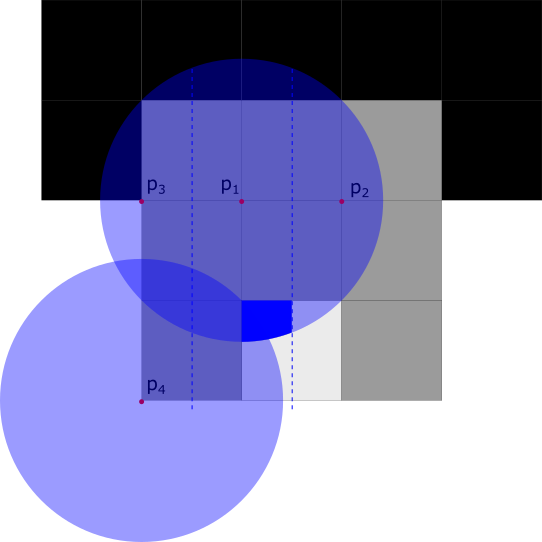}
\hfill
\includegraphics[scale=0.4]{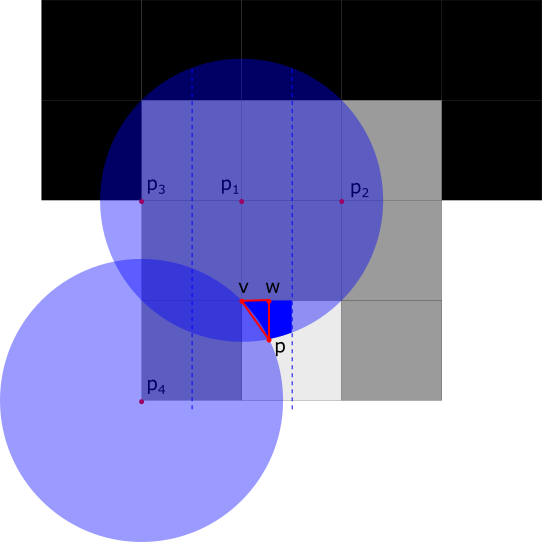}
\caption{Left: Since $p_1$ is white, it must lie in a white $\sqrt{2}$-ball $B$ centred somewhere inside $B_{\sqrt{2}}(p_1)$ (the blue circle. Since $B$ must contain neither $p_2$ nor $p_3$, the centre of $B$ must lie between the two dashed lines. Right: A point in the left part of the blue set $T'$ must belong to the triangle with vertices $p$, $v$ and $w$.}
\label{FigImpossible3x3B}
\end{figure}

Since $p_1$ is white, it is contained in some white $\sqrt{2}$-ball $B$. The centre of $B$ lies somewhere inside $B_{\sqrt{2}}(p_1)$. Since $B$ cannot contain the black corner $p_3$ (see Figure \ref{FigImpossible3x3B} left), its center must lie closer to $p_1$ than to $p_2$, hence it must lie to the right of the vertical line midway between the two points. Likewise, $B$ cannot contain both $p_1$ and $p_2$ without also containing the entire line between them, so the centre of $B$ must also lie closer to $p_1$ than to $p_2$, that is, to the left of the vertical line midway between the two points.

Finally, the centre of $B$ can only belong to a white pixel. Hence it must belong to the bright blue part of the white pixel in Figure \ref{FigImpossible3x3B}, left. Let us call this set $T$.

Next, consider the bottom left grey pixel $C$, and let $p_4$ be its lower left corner. Since $T$ is a part of the upper left quarter of the white pixel next to it, any point in $T$ is further from $p_4$ than from any of the other three corners of $C$. Hence if $B$ contained $p_4$, it would also contain all the corners of $C$ that are closer to its centre, and hence it would contain all of $C$ which would then not be grey. So the centre of $B$ must lie further away from $p_4$ than $\sqrt{2}$, hence outside the ball $B_{\sqrt{2}}(p_4)$. Let $T'=T\backslash B_{\sqrt{2}}(p_4)$, the blue set in Figure \ref{FigImpossible3x3B}, right.

By calculating the intersection $p$ between the boundaries of $B_{\sqrt{2}}(p_1)$ and $B_{\sqrt{2}}(p_4)$ inside the white pixel, we find that they intersect in a point that is a distance $\frac{1}{10}(2\sqrt{15}-5)\approx 0.27$ from the left edge of the white pixel and a distance $\sqrt{3/20}\approx 0.39$ from the top edge of the white pixel.

Let $m$ be the midpoint of the line between $p_1$ and $p_2$. A calculation shows that $d(p,m)<\sqrt{2}$, so $m\in B_{\sqrt{2}}(p)$. Any point of $T'$ to the right of $p$ is closer to $m$ than $p$ is, so any ball centred here must also contain $m$.

Consider a point in $T'$ left of $p$. Such a point must be contained in the triangle with corners $p$, $v$ and $w$ as seen in Figure \ref{FigImpossible3x3B}, right. Here $v$ is the upper left vertex of the white pixel, and $w$ is the point on the edge of the white pixel that is directly above $p$. A calculation shows that $m\in B_{\sqrt{2}}(p)\cap B_{\sqrt{2}}(v)\cap B_{\sqrt{2}}(w)$, which by Lemma \ref{LemCircleIntersections} means that $m$ belongs to $B_{\sqrt{2}}(x)$ for any $x$ in the left part of $T'$. Since the same was true for any point right of $p$, $m$ belongs to $B_{\sqrt{2}}(x)$ for any $x\in T'$. But since $p_1$ also belongs to $B_{\sqrt{2}}(x)$ for any $x$ in $T'$, any white ball containing $p_1$ also contains $m$, and hence the line segment from $p_1$ to $m$.

Repeating this argument for $p_2$, any white ball containing $p_2$ also contains $m$, and hence it contains the entire line segment from $m$ to $p_2$.

But then each point on the line segment from $p_1$ to $p_2$ is contained in a white ball - a contradiction.
\end{proof}

\begin{proof}[Proof of Theorem \ref{ThmConfigurationsUsualResolution}]
Combining Lemmas \ref{LemSkalfarvessorte1}, \ref{LemTwoIntersections}, \ref{LemYieldsTwoIntersections},  \ref{LemNo9greys}, \ref{Thm5greys}, \ref{LemImpossible2x3} and \ref{LemTrickyConfiguration}, we get the result with the exception of the configuration located at $(23,17)$ in Figure \ref{FigMany3x3Configurations}. But this configuration is also impossible: Let $C$ be the grey centre pixel. $C$ must contain some boundary point $p\in \del X$, so there would have to be a white ball of radius $r$ tangent to $\del X$ at $p$, and in particular, there would have to be a white ball of radius $1$ with $p$ on its boundary. But any point that is closer than $1$ to a point $p$ in $C$ would either have to belong to either $C$ or to one of the neighbouring black pixels, so the same must be true for the centre of the white $1$-ball tangent to $p$. If the centre belonged to a black pixel it would not be black, and if it belonged to $C$, the white ball of radius $d$ would contain some set of interior points of a pixel adjacent to $C$ - but these are all black. Hence we reach a contradiction.
\end{proof}

\begin{proof}[Proof of Lemma \ref{LemLeftoverconfiguration}]
  Look at the centre point $c$ of the $4\times4$ pixels. Suppose
  it is white (the case where it is black is symmetric). Then one of the grey
  pixels having $c$ as a vertex has only white vertices (in the
  figure, it would be the lower left pixel). Call this pixel $A$.

  Since $A$ s grey, the boundary $\del X$ must intersect one of its
  edges, and since all of its corners are white, an edge intersected
  by $\del X$ must be intersected at least twice. Note that only the
  edges of $A$ that are shared with another grey pixel can be
  intersected by $\del X$. But then by Lemma
  \ref{LemTwoIntersections}, one of the grey pixels in the figure
  would have to be non-grey - a contradiction. So this configuration
  cannot occur.
\end{proof}

\begin{proof}[Proof of Lemma \ref{Lem2x4Configuration}]
  The proof follows from Lemma \ref{Lem2x3grey}: Look at the two
  pixels in the column to the right of the configuration (the red ones
  in Figure \ref{Fig2x4Configuration}, right). These cannot both be
  grey, since that would violate Theorem \ref{Lem2x3grey}, so at least
  one of them must have another colour. Say that one of them is black
  (the other case is symmetric). Depending on which one of the red
  pixels is black, some part of Lemma \ref{LemSkalfarvessorte2} tells
  us that the $2\times4$ configuration in Figure
  \ref{Fig2x4Configuration} must have more black pixels than what is
  the case - a contradiction. So the configuration cannot occur.
\end{proof}

\begin{proof}[Proof of Lemma \ref{Lem4Intersections2x2}]
  Let $C$ denote the $2\times 2$ configuration of grey pixels. By a proof copying the proof of Lemma \ref{LemTwoIntersections}, if any edge of $C$ is intersected
  by $\del X$ multiple times, then one of the pixels in $C$ would not be grey - a contradiction. So if $\del X$ intersects all edges of
  $C$, it only intersects each edge once. Hence $C$ has two black
  vertices on one diagonal and two white vertices on the other. Let
  $L$ be the line connecting the black vertices and $M$ the line
  connecting the white vertices.

  There is a black path $\rho_X(L)$ in $X$ connecting the two
  black corners of the pixel. By Corollary \ref{CorPathInSpindle} and Lemma \ref{LemSpindleIsIntersection}, this path must belong to $B_{\sqrt{2}}(p)$, where $p$ is the centre of $C$.

  Similarly, there is a path $\rho_{X^C}(M)$ in $X^C\cup\del X$
  connecting the two white vertices and contained in $B_{\sqrt{2}}(p)$. If
  $\rho_{X^C}(M)$ contains points of $\del X$, we may push these
  points a little along the normal vector field of $\del X$ to get a
  path $\tilde{\rho}$ in $X^C$ connecting the white vertices of
  $C$. This alteration can be made in the interior of $B_{\sqrt{2}}(p)$ since the endpoints of $M$ are not boundary points, hence $\tilde{\rho}$ is also contained in $C$.

  But then we have a black path in $B_{\sqrt{2}}(p)$ separating the white vertices
  of $C$, and a white path in $B_{\sqrt{2}}(p)$ connecting them. This means that the
  two paths must necessarily intersect each other in a point that must
  be both black and white -- a contradiction. So $\del X$ cannot
  intersect all four edges of $C$.
\end{proof}

\begin{proof}[Proof of Theorem \ref{Thm2x2}]
Combining Lemmas \ref{LemLeftoverconfiguration}, \ref{Lem2x4Configuration} and \ref{Lem4Intersections2x2} yields most of the result. The only configuration remaining is the one centred at $(45,22)$ in Figure \ref{FigKonfigurationer}. But this is also not possible: If it where, the middle grey pixels would contain some boundary point $p\in\del X$. Then there would be a white $\sqrt{2}$-ball with $p$ in its boundary, and such a ball would either be centred inside the $4\times4$ pixels of the configuration, or in one of the pixels neighbouring the configuration. Since the pixel containing the centre of the white ball is white itself, the white ball cannot be centred inside the configuration. But it also cannot be centred in one of the pixels neighbouring the configuration, since this would mean that a white pixel and a black one were sharing boundary points, which is against our assumption.
\end{proof}
\newpage
\bibliographystyle{plain}
\bibliography{ms}
\end{document}